\documentclass[journal,comsoc]{IEEEtran}

\usepackage[T1]{fontenc}

\usepackage{cite}
\usepackage{afterpage}

\ifCLASSINFOpdf
  \usepackage[pdftex]{graphicx}
  \usepackage{epstopdf}
\else
  \usepackage[dvips]{graphicx}
\fi

\usepackage{amsmath}
\allowdisplaybreaks

\usepackage{amsthm}

\usepackage[cmintegrals]{newtxmath}

\usepackage{algorithmic}
\usepackage[algoruled, linesnumbered]{algorithm2e} 

\usepackage{array}

\ifCLASSOPTIONcompsoc
  \usepackage[caption=false,font=normalsize,labelfont=sf,textfont=sf]{subfig}
\else
\usepackage[caption=false,font=footnotesize]{subfig}
\fi

\usepackage{fixltx2e}

\usepackage{stfloats}

\usepackage{url}

\usepackage{xpatch}
\usepackage{tikz}
\usetikzlibrary{patterns}
\usepackage{setspace}
\usepackage{amsthm}
\allowdisplaybreaks
\usepackage[normalem]{ulem}
\usepackage{color}
\usepackage{footmisc}
\newtheorem{theorem}{Theorem}[section]

\newtheorem{remark}{Remark}[section]
\allowdisplaybreaks

\begin{document}
\title{Periodic Analog Channel Estimation Aided Beamforming for Massive MIMO Systems}
\author{Vishnu~V.~Ratnam,~\IEEEmembership{Student~Member,~IEEE,}
and Andreas~F.~Molisch,~\IEEEmembership{Fellow,~IEEE,}
\thanks{V. V. Ratnam and A. F. Molisch are with the Ming Hsieh Department of Electrical and Computer Engineering, University of Southern California, Los Angeles, CA, 90089 USA (e-mail: \{ratnam, molisch\}@usc.edu). This work was supported by the National Science Foundation under project CIF-1618078.} 
}

\IEEEtitleabstractindextext{%
\begin{abstract}
Analog beamforming is an attractive and cost-effective solution to exploit the benefits of massive multiple-input-multiple-output systems, by requiring only one up/down-conversion chain. However, the presence of only one chain imposes a significant overhead in estimating the channel state information required for beamforming, when conventional digital channel estimation (CE) approaches are used. As an alternative, this paper proposes a novel CE technique, called periodic analog CE (PACE), that can be performed by analog hardware. By avoiding digital processing, the estimation overhead is significantly lowered and does not scale with number of antennas. PACE involves periodic transmission of a sinusoidal reference signal by the transmitter, estimation of its amplitude and phase at each receive antenna via analog hardware, and using these estimates for beamforming. To enable such non-trivial operation, two reference tone recovery techniques and a novel receiver architecture for PACE are proposed and analyzed, both theoretically and via simulations. Results suggest that in sparse, wide-band channels and above a certain signal-to-noise ratio, PACE aided beamforming suffers only a small loss in beamforming gain and enjoys a much lower CE overhead, in comparison to conventional approaches. Benefits of using PACE aided beamforming during the initial access phase are also discussed.
\end{abstract}

}

\maketitle

\IEEEdisplaynontitleabstractindextext

\IEEEpeerreviewmaketitle

\section{Introduction} \label{sec_intro}
Massive Multiple-input-multiple-output (MIMO) systems, enabled by using antenna arrays with many elements at the transmitter (TX) and/or receiver (RX), promise large beamforming gains and improved spectral efficiency, and are therefore a key focus area for 5G systems research and development \cite{Marzetta, Keytech_5G}. Such massive antenna arrays, while also beneficial at sub-$6$ GHz frequencies, are essential at the higher millimeter-wave (mm-wave) frequencies to compensate for the large channel attenuation. However, despite their numerous benefits, full complexity massive MIMO architectures suffer from increased hardware cost and energy consumption. This is because, though the antenna elements are affordable, the corresponding up/down-conversion chains - which include circuit components such as analog-to-digital converters  and digital-to-analog converters - are both expensive and power hungry \cite{Murmann_ADC_compiled}. A popular solution to reduce this implementation cost is hybrid beamforming \cite{Heath2016, Molisch_HP_mag}, where the large antenna array is connected to a small number of up/down-conversion chains via power-efficient and cost-effective analog hardware, such as, phase-shifters. By using such analog hardware to focus power into the dominant channel directions, hybrid beamforming exploits the directional nature of wireless channels to minimize loss in system performance. In this paper, we focus on a special case of hybrid beamforming with one up/down-conversion chain (for the in-phase and quadrature-phase components each), referred to as analog beamforming.

A major challenge with analog beamforming (and also hybrid beamforming in general) is the acquisition of the channel state information (CSI) required for beamforming at the TX and RX. In narrow-band (i.e., frequency-flat fading) systems \cite{Molisch_VarPhaseShift, Ayach_iCSI, Xu_iCSI, Alkhateeb2015, Sohrabi2016, Vishnu_ICC2017, Ratnam_HBwS_jrnl}, the required CSI usually involves instantaneous channel parameters (iCSI), while in wide-band systems \cite{Sudarshan, Adhikary_JSDM, Liu2014, Caire2017, Zheda2017, Park2017} average channel parameters (aCSI) are used for designing the analog beamformer. Here aCSI refers to channel parameters that remain constant over a wide time-frequency range, such as the spatial correlation matrices, while iCSI are parameters that change faster. 
In either scenario, the required CSI can be obtained by transmitting known signals (pilots) and performing channel estimation (CE) at the RX within each CSI coherence time, i.e., period over which CSI remains constant.\footnote{Required CSI at the TX is obtained either via CE on the reverse link, or via CSI feedback from RX. \label{note3}} 
Since all RX antennas share one down-conversion chain, multiple temporal pilot transmissions are required for performing such CE \cite{Alkhateeb2014, Jeong2015, Vishnu_Globecom, Caire2017}. 
As an illustration, exhaustive CE approaches \cite{Jeong2015} require ${\rm O}(M_{\rm tx} M_{\rm rx})$ pilots, where $M_{\rm tx}, M_{\rm rx}$ are the number of TX and RX antennas, respectively and ${\rm O}(\cdot)$ represents the scaling behavior in big oh notation. 
Such a large pilot overhead may consume a significant portion of the time-frequency resources when the CSI coherence time is short, such as in vehicle-to-vehicle channels, in systems using narrow TX/RX beams, e.g., massive MIMO, or in channels with large carrier frequencies and high blocking probabilities, e.g., at mm-wave frequencies \cite{Akdeniz2014}. 
The overhead also increases system latency and makes the initial access\footnote{Initial access refers to the phase wherein, a user equipment and base-station discover each other, synchronize, and coordinate to initiate communication.} (IA) procedure very cumbersome \cite{Barati2015, Li2016, Giordani2016_iter}. Several fast CE approaches have therefore been suggested to reduce the pilot overhead, which are discussed below assuming $M_{\rm tx}=1$ for convenience.\footnote{For $M_{\rm tx} > 1$, the pilot overhead increases further, either multiplicatively or additively, by a function of $M_{\rm tx}$, determined by the CE algorithm used at the TX. \label{note1}}
Side information aided narrow-band CE approaches utilize channel statistics and temporal correlation to reduce the iCSI pilot overhead \cite{Love2014, Adhikary2014, Ratnam_HBwS_jrnl, Vishnu_Globecom}. Compressed sensing based approaches \cite{Chi2013, Alkhateeb2014, Park2016_asilomar, Lee2016} exploit the sparse nature of the massive MIMO channels to reduce the pilots up to ${\rm O}[L \log(M_{\rm rx}/L)]$ per CSI coherence time, where $L$ is the channel sparsity level. Iterative angular domain CE performs beam sweeping at the RX with progressively narrower search beams to find a good beam direction with ${\rm O}(\log M_{\rm rx})$ pilots \cite{kim2014fast, Desai2014, Giordani2016_iter}. Approaches that utilize side information to improve iterative angular domain CE \cite{Giordani2016_context, Devoti2016} or perform angle domain tracking \cite{Gao2017, Li2017} have also been considered. Sparse ruler based approaches exploit the possible Toeplitz structure of the spatial correlation matrix to reduce pilots to ${\rm O}(\sqrt{M_{\rm rx}})$ per CSI coherence time \cite{Pal2010, Pal2011, Romero2013, Rial2015_icassp, Caire2017}. Since the overhead still scales with $M_{\rm rx}$\footref{note1}, these approaches are only partially successful in reducing the pilot overhead. Furthermore, some of these CE approaches may not be applicable for IA since they would require the timing and frequency synchronization \cite{Nasir2016, Meng2017} to be performed without the TX/RX beamforming gain, which may be difficult at the low signal-to-noise ratio (SNR) and high phase noise (i.e., random fluctuations of the instantaneous oscillator frequency) levels expected in mm-wave systems. Some of these CE approaches also require the channel to remain static during the re-transmissions and are only applicable for certain antenna configurations and/or channel models. 
Finally, to reduce the impact of the transient effects of analog hardware on CE \cite{Sands2002}, the multiple pilots may have to be spaced sufficiently far apart\cite{Venugopal2017}, thus potentially increasing the latency. 

The main reason for the pilot overhead is that conventional CE approaches require processing in the digital domain, thus having to time-share the down-conversion chain across the antennas. Inspired by ultra-wideband transmit reference schemes \cite{Hoctor2002, Goeckel2007, Ratnam_Globecom17} and legacy adaptive antenna array techniques \cite{Breese1961, Ghose1964, Thompson1976, Golliday1982}, our recent conference papers \cite{Ratnam_ICC2018, Ratnam_Globecom2018}, explore a different novel approach that enables CE without digital processing. In this approach, the TX transmits a reference sinusoidal tone simultaneously with the data. The received reference signals (including both amplitude and phase) are then recovered at each RX antenna via analog hardware and are utilized as a homodyne combining filter for the data. In essence, \cite{Ratnam_ICC2018, Ratnam_Globecom2018} show that a maximal ratio combining (MRC) beamformer built for a reference frequency also provides a good, albeit sub-optimal, beamforming gain at other frequencies in a sparse scattering, wide-band channel. This is because, although they experience frequency selective fading, such channels exhibit a strong coupling across frequency. 
Since recovering a reference sinusoidal signal, or equivalently estimating its amplitude and phase, is significantly simpler than conventional CE, it can be performed at each RX antenna by analog hardware such as phase locked loops. Thus, by avoiding digital CE, this scheme allows RX beamforming without pilot re-transmissions. We shall henceforth refer to this type of amplitude and phase estimation as \emph{analog channel estimation} (ACE). 
Note that due to the limited capabilities of analog hardware and the low SNR before beamforming, performing ACE and exploring new ACE techniques is non-trivial. In the original design in \cite{Ratnam_ICC2018}, the reference has to be transmitted continuously, to enable its recovery at the RX. 
While this design reduces the estimation overhead and avoids phase-shifters, it requires $M_{\rm rx}$ carrier recovery circuits which may add to the cost and power consumption of the RX. Furthermore, the continuous recovery of the reference tone is an overkill, and may cause some wastage in the transmit power and spectral efficiency. 
In \cite{Ratnam_Globecom2018}, a non-coherent variant of \cite{Ratnam_ICC2018} is explored that avoids recovery circuits but at the expense of $50\%$ bandwidth efficiency reduction. The current paper therefore proposes a different ACE scheme, referred to as periodic ACE (PACE), where the reference is transmitted judiciously, and its amplitude and phase are explicitly estimated to drive an RX phase shifter array. Unlike \cite{Ratnam_ICC2018}, PACE requires one carrier recovery circuit and $M_{\rm rx}$ phase shifters (see Fig.~\ref{Fig_block_diag}) and can support both homo/heter-dyne reception. 

In PACE, the TX transmits a reference tone at a known frequency during each periodic RX beamformer update phase. One carrier recovery circuit, involving phase-locked loops (PLLs), is used to recover the reference tone from one or more antennas, as shown in Fig.~\ref{Fig_block_diag}. This recovered reference tone, and its quadrature component, are then used to estimate the phase off-set and amplitude of the received reference tone at each RX antenna, via a bank of `filter, sample and hold' circuits (represented as integrators in Fig.~\ref{Fig_block_diag}). As shall be shown, these estimates are proportional to the channel response at the reference frequency. 
These estimates are used to control an array of variable gain phase-shifters, which generate the RX analog beam. During the data transmission phase, the wide-band received data signals pass through these phase-shifters, are summed and processed similar to conventional analog beamforming. As the phase and amplitude estimation is done in the analog domain, ${\rm O}(1)$ pilots are sufficient to update the RX beamformer. 
Additionally, the power from multiple channel MPCs is accumulated by this approach, increasing the system diversity against MPC blocking. Furthermore, the same variable gain phase-shifts can also be used for transmit beamforming on the reverse link. Finally, by providing an option for digitally controlling the inputs to the phase-shifters, the proposed architecture can also support conventional beamforming approaches.  

On the flip side, PACE requires some additional analog hardware components, such as mixers and filters, in comparison to conventional digital CE. Additionally, the accumulation of power from multiple MPCs may cause frequency selective fading in a wide-band scenario, which can degrade performance. Finally, the proposed approach in its current suggested form does not support reception of multiple spatial data streams and can only be used for beamforming at one end of a communication link. This architecture is therefore more suitable for use at the user equipment (UEs). The possible extensions to multiple spatial stream reception shall be explored in future work. While the proposed architecture is also applicable in narrow-band scenarios, in this paper we shall focus on the analysis of a wide-band scenario where the repetition interval of PACE and beamformer update is of the order of aCSI coherence time, i.e. time over which the aCSI stays approximately constant (also called stationarity time in some literature). 

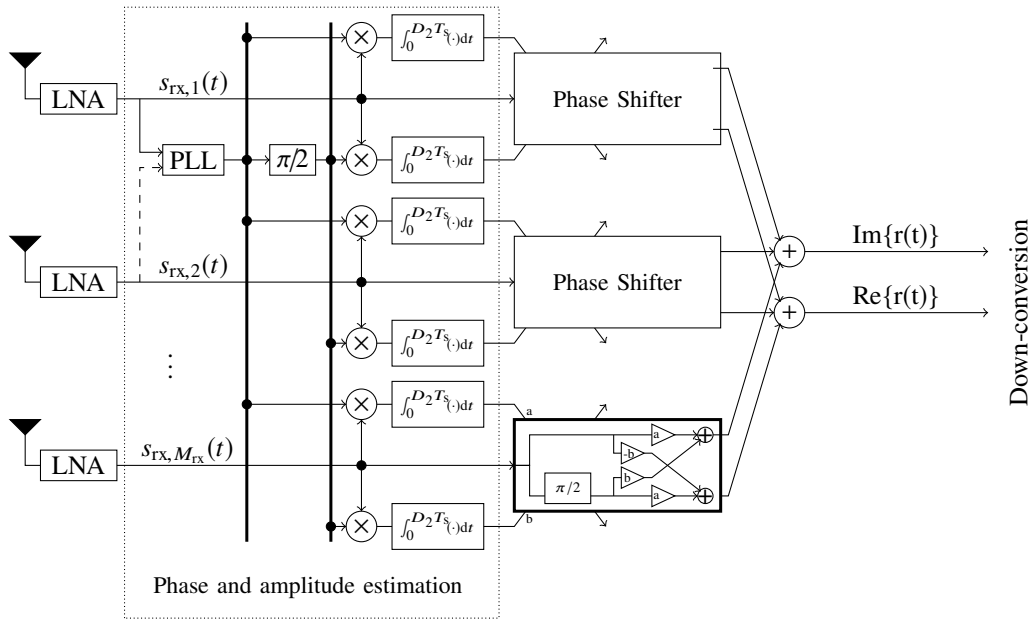
\begin{figure*}[h] 
\setlength{\unitlength}{0.08in} 
\centering 
\begin{tikzpicture}[x=0.08in,y=0.08in] 
\draw[line width = 0.4 mm] (14.5,0) -- (14.5,34); \draw[line width = 0.4 mm] (20,0) -- (20,34); 
\filldraw[color=black, fill=black](0,7) (0,7) -- (1,8) -- (-1,8) -- cycle;
\draw (0,7) -- (0,5) -- (1,5); \draw(1,4) rectangle (6,6) node[pos=.5] {LNA}; \draw[->] (6,5) -- (32,5); \draw(10.6,6) node{$s_{{\rm rx},M_{\rm rx}}(t)$}; \draw[->] (22,5) -- (22,2);
\filldraw[fill=black] (20,1) circle(0.3); \draw[->] (20,1) -- (21,1); 
\draw (22,1) circle(1); \draw(22,1) node{$\times$}; \draw (23,1) -- (24,1); \draw(24,-0.5) rectangle (30,2.5) node[pos=.5] {\tiny $\int_{0}^{D_2 T_{\rm s}} \!\!(\cdot) {\rm d}t$}; 
\filldraw[fill=black] (14.5,9) circle(0.3); \draw[->] (14.5,9) -- (21,9); \draw[->] (22,5) -- (22,8); \filldraw[fill=black] (22,5) circle(0.3);
\draw (22,9) circle(1); \draw(22,9) node{$\times$}; \draw (23,9) -- (24,9); \draw(24,7.5) rectangle (30,10.5) node[pos=.5] {\tiny $\int_{0}^{D_2 T_{\rm s}} \!\! (\cdot) {\rm d}t$}; 
\draw[->] (30,1) -- (32,1) -- (38,9); \draw[->] (30,9) -- (32,9) -- (38,1); 
\draw[fill = white, line width = 0.4 mm] (32,2) rectangle (45.5,8); \draw(33,8.5) node{\tiny a}; \draw(33,1.5) node{\tiny b};
\draw (32,5) -- (33,5) -- (33,7) -- (41,7); 
\draw (41,6.3) -- (41,7.7) -- (42.5,7) -- cycle; \draw(41.3,7) node{\tiny a}; \draw (38.5,7) -- (38.5, 5.8) -- (39, 5.8); \draw (39,5.1) -- (39,6.5) -- (40.5,5.8) -- cycle; \draw(39.5,5.8) node{\tiny -b};
\draw[->] (42.5,7) -- (44,7); \draw(44.5,7) circle(0.5); \draw(44.5,7) node{+}; \draw[->] (40.5,4.2) -- (41,4.2) -- (44.2,6.7); \draw[->] (45,7) -- (46,7) -- (49.4,18.3);  
\draw (32,5) -- (33,5) -- (33,3) -- (34,3); \draw(34,2.5) rectangle (37,4.5) node[pos=.5] {\tiny $\pi/2$}; \draw(37,3) -- (41,3);
\draw (41, 2.3) -- (41,3.7) -- (42.5,3) -- cycle; \draw(41.3,3) node{\tiny a}; \draw(38.5,3) -- (38.5, 4.2) -- (39, 4.2); \draw (39,4.9) -- (39,3.5) -- (40.5,4.2) -- cycle; \draw(39.4,4.2) node{\tiny b}; \draw[->] (42.5,3) -- (44,3); \draw(44.5,3) circle(0.5); \draw(44.5,3) node{+}; \draw[->] (40.5,5.8) -- (41,5.8) -- (44.2,3.3); \draw[->] (45.5,3) -- (46,3) -- (49.4,14.3); 
\draw(9.5,12) node{\vdots}; 
\filldraw[color=black, fill=black](0,19) (0,19) -- (1,20) -- (-1,20) -- cycle;
\draw (0,19) -- (0,17) -- (1,17); \draw(1,16) rectangle (6,18) node[pos=.5] {LNA}; \draw[->] (6,17) -- (32,17); \draw(11,18) node{$s_{{\rm rx},2}(t)$};  \draw[->] (22,17) -- (22,14); 
\filldraw[fill=black] (20,13) circle(0.3); \draw[->] (20,13) -- (21,13); \draw (22,13) circle(1); \draw(22,13) node{$\times$}; \draw (23,13) -- (24,13); \draw(24,11.5) rectangle (30,14.5) node[pos=.5] {\tiny $\int_{0}^{D_2T_{\rm s}} \!\! (\cdot) {\rm d}t$}; 
\filldraw[fill=black] (14.5,21) circle(0.3); \draw[->] (14.5,21) -- (21,21); \draw[->] (22,17) -- (22,20); \filldraw[fill=black] (22,17) circle(0.3); \draw (22,21) circle(1); \draw(22,21) node{$\times$}; \draw (23,21) -- (24,21); \draw(24,19.5) rectangle (30,22.5) node[pos=.5] {\tiny $\int_{0}^{D_2 T_{\rm s}} \!\!(\cdot) {\rm d}t$}; 
\draw[->] (30,13) -- (32,13) -- (38,21); \draw[->] (30,21) -- (32,21) -- (38,13); \draw[fill = white](32,14) rectangle (45.5,20) node[pos=.5] {\small Phase Shifter}; \draw[->] (45.5,15) -- (49,15); \draw[->] (45.5,19) -- (49,19);
\draw[dashed,->] (7.5,17) -- (7.5,24.5) -- (9,24.5);
\draw[->] (7.5,29) -- (7.5,25.5) -- (9,25.5); \draw(9,24) rectangle (13,26) node[pos=.5] {PLL}; \draw[->] (13,25) -- (16,25); \draw(16,24) rectangle (19,26) node[pos=.5] {$\pi \!/\! 2$}; \draw(19,25) -- (20,25); \filldraw[fill=black] (14.5,25) circle(0.3); 
\filldraw[color=black, fill=black](0,31) (0,31) -- (1,32) -- (-1,32) -- cycle;
\draw (0,31) -- (0,29) -- (1,29); \draw(1,28) rectangle (6,30) node[pos=.5] {LNA}; \draw[->] (6,29) -- (32,29); \draw(11,30) node{$s_{{\rm rx},1}(t)$}; \filldraw[fill=black] (20,25) circle(0.3); \draw[->] (20,25) -- (21,25); \draw (22,25) circle(1); \draw(22,25) node{$\times$}; \draw (23,25) -- (24,25);  \draw(24,23.5) rectangle (30,26.5) node[pos=.5] {\tiny $\int_{0}^{D_2 T_{\rm s}} \!\!(\cdot) {\rm d}t$}; 
\filldraw[fill=black] (14.5,33) circle(0.3); \draw[->] (14.5,33) -- (21,33); \draw[->] (22,29) -- (22,32); \draw[->] (22,29) -- (22,26); \filldraw[fill=black] (22,29) circle(0.3); \draw (22,33) circle(1); \draw(22,33) node{$\times$}; \draw (23,33) -- (24,33); \draw(24,31.5) rectangle (30,34.5) node[pos=.5] {\tiny $\int_{0}^{D_2 T_{\rm s}} \!\! (\cdot) {\rm d}t$}; 
\draw[->] (30,25) -- (32,25) -- (38,33); \draw[->] (30,33) -- (32,33) -- (38,25); \draw[fill=white](32,26) rectangle (45.5,32) node[pos=0.5]{\small Phase Shifter}; 
\draw[->] (45,31) -- (46,31) -- (49.4,19.7); \draw[->] (45,27) -- (46,27) -- (49.4,15.7); 
\draw(50,19) circle(1); \draw(50,19) node{$+$}; \draw[->] (51,19) -- (63,19); \draw(57,20) node{$\rm{Im}\{r(t)\}$};
\draw(50,15) circle(1); \draw(50,15) node{$+$}; \draw[->] (51,15) -- (63,15); \draw(57,16) node{$\rm{Re}\{r(t)\}$};
\draw[densely dotted] (6.5,-5) rectangle (31,35); \draw(18.5,-3) node{\small Phase and amplitude estimation}; 
\draw(65,15) node[rotate=90] {Down-conversion};
\end{tikzpicture}
\caption{Block diagram of an RX with analog beamforming enabled via periodic analog channel estimation.} 
\label{Fig_block_diag} 
\end{figure*} 
 
The contributions of this paper are as follows:
\begin{enumerate}
\item We propose a novel transmission technique, namely PACE, and a corresponding RX architecture that enable RX analog beamforming with low CE overhead. 
\item To enable the RX operation, we also explore two novel reference recovery circuits. These circuits are non-linear, making their analysis non-trivial. We provide an approximate analysis of their phase-noise and the resulting performance that is tight in the high SNR regime. 
\item We analytically characterize the achievable system throughput with PACE aided beamforming in a wide-band channel. 
\item Simulations with practically relevant channel models are used to support the analytical results and compare performance to existing schemes.
\end{enumerate}
The organization of the paper is as follows: the system model is presented Section \ref{sec_chan_model}; two designs for PACE and their respective noise analysis is presented in Section \ref{sec_analog_phase_amp_est}; the system performance with PACE aided beamforming is characterized in Section \ref{sec_data_transmission}; the advantages of PACE for transmit beamforming and during the IA phase are discussed in Section \ref{sec_IA_aCSI_at_BS}; simulations results are presented in Section \ref{sec_sim_results} and finally conclusions are in Section \ref{sec_conclusions}. 

\textbf{Notation:} scalars are represented by light-case letters; vectors by bold-case letters; and sets by calligraphic letters. Additionally, ${\rm j} = \sqrt{-1}$, $a^*$ is the complex conjugate of a complex scalar $a$, $|\mathbf{a}|$ represents the $\ell_2$-norm of a vector $\mathbf{a}$ and ${\mathbf{A}}^{\dag}$ is the conjugate transpose of a complex matrix $\mathbf{A}$. Finally, $\mathbb{E}\{\}$ represents the expectation operator, $\otimes$ represents the Kronecker product, $\stackrel{\rm d}{=}$ represents equality in distribution, $\mathrm{Re}\{\cdot\}$/$\mathrm{Im}\{\cdot\}$ refer to the real/imaginary component, respectively, $\mathcal{CN}(\mathbf{a},\mathbf{B})$ represents a circularly symmetric complex Gaussian vector with mean $\mathbf{a}$ and covariance matrix $\mathbf{B}$, $\mathrm{Exp}\{a\}$ represents an exponential distribution with mean $a$ and $\mathrm{Uni}\{a,b\}$ represents a uniform distribution in range $[a,b]$. 

\section{General Assumptions and System model} \label{sec_chan_model}
We consider the downlink of a single-cell MIMO system, wherein one base station (BS) with $M_{\rm tx}$ antennas transmits to several UEs with $M_{\rm rx}$ antennas each.  Since  focus is on the downlink, we shall use abbreviations BS \& TX and UE \& RX interchangeably. Each UE is assumed to have one up/down-conversion chain, while no assumptions are made regarding the BS architecture. 
\begin{figure}[!htb]
\centering
\includegraphics[width= 0.48\textwidth]{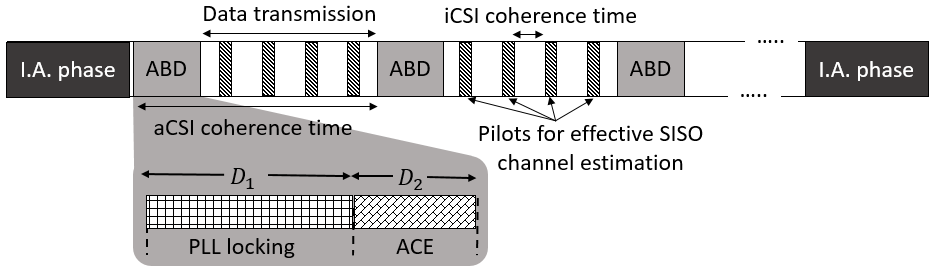}
\caption{An illustrative transmission block structure for the PACE scheme.}
\label{Fig_tx_protocol}
\end{figure}
Here we assume the communication between the BS and UEs to involve three important phases: (i) initial access (IA) - where the BS and UEs find each other, timing/frequency synchronization is attained and spectral resources are allocated; (ii) analog beamformer design - where the BS and UEs obtain the required aCSI to update the analog precoding/combining beams; and (iii) data transmission. 
The relative time scale of these phases are illustrated in Fig.~\ref{Fig_tx_protocol}. Through most of this paper (Sections \ref{sec_chan_model}-\ref{sec_data_transmission}), we assume that the IA and beamformer design at the BS are already achieved, and we mainly focus on the beamformer design phase at the UE and the data transmission phase. Therefore we assume perfect timing and frequency synchronization between the BS and UE, and assume that the TX beamforming has been pre-designed based on aCSI at the BS. Later in Section \ref{sec_IA_aCSI_at_BS}, we also briefly discuss how aCSI can be acquired at the BS, how IA can be performed and how the use of PACE can be advantageous in those phases. 

The BS transmits one spatial data-stream to each scheduled UE, and all such scheduled UEs are served simultaneously via spatial multiplexing. Furthermore, the data to the UEs is assumed to be transmitted via orthogonal precoding beams, such that, there is no inter-user interference.\footnote{This type of precoding is possible by avoiding transmission to the scatterers common to multiple scheduled UEs \cite{Adhikary2014}.} Under these assumptions and given transmit precoding beams and power allocation, we shall restrict the analysis to one representative UE without loss of generality. For convenience, we shall also assume the use of noise-less and perfectly linear antennas, filters, amplifiers and mixers at both the BS and UE. An analysis including the non-linear effects of these components is beyond the scope of this paper. The BS transmits orthogonal frequency division multiplexing (OFDM) symbols with $K$ sub-carriers, indexed as $\mathcal{K} = \{-K_1,...,K_2-1,K_2\}$ with $K_1+K_2+1 = K$, to this representative UE.\footnote{While the proposed PACE technique is also applicable to single carrier transmission, a detailed analysis of the same is beyond the scope of this paper.} 
The BS transmits two kinds of symbols: reference symbols and data symbols. In a reference symbol, only a reference tone, i.e., a sinusoidal signal with a pre-determined frequency known both to the BS and UE, is transmitted on the $0$-th subcarrier, and the remaining sub-carriers are all empty. 
On the other hand, in a data symbol all the $K$ sub-carriers are used for data transmission.\footnote{In an actual implementation the data symbols may have may also have null and pilot sub-carriers, but we ignore them here for simplicity.} The purpose of the reference symbols is to aid PACE and beamformer design at the RX, as shall be explained later. Since the BS can afford an accurate oscillator, we shall assume that the BS suffers negligible phase noise. The $M_{\rm tx} \times 1$ \emph{complex equivalent} transmit signal for the $0$-th symbol, if it is a reference or data symbol, respectively, can then be expressed as:
\begin{subequations} \label{eqn_tx_signal}
\begin{eqnarray} 
\tilde{\mathbf{s}}^{(\rm r)}_{\rm tx}(t) &=& \sqrt{\frac{2}{T_{\rm cs}}} \mathbf{t} \sqrt{E^{(\rm r)}} e^{{\rm j} 2 \pi f_{\rm c} t } \label{eqn_tx_signal_ref} \\
\tilde{\mathbf{s}}^{(\rm d)}_{\rm tx}(t) &=& \sqrt{\frac{2}{T_{\rm cs}}} \mathbf{t}\bigg[ \sum_{k \in \mathcal{K}} x^{(\rm d)}_k e^{{\rm j} 2 \pi f_k t} \bigg] e^{{\rm j} 2 \pi f_{\rm c} t }, \label{eqn_tx_signal_data}
\end{eqnarray}
\end{subequations}
for $-T_{\rm cp} \leq t \leq T_{\rm s}$, where $\mathbf{t}$ is the $M_{\rm tx}\times 1$ unit-norm TX beamforming vector for this UE with $|\mathbf{t}|=1$, $x^{(\rm d)}_k$ is the data signal at the $k$-th OFDM sub-carrier, ${\rm j} = \sqrt{-1}$, $f_{\rm c}$ is the carrier/reference frequency, $f_k = k/T_{\rm s}$ represents the frequency offset of the $k$-th sub-carrier, $T_{\rm cs}= T_{\rm cp} + T_{\rm s}$ and $T_{\rm s}, T_{\rm cp}$ are the symbol duration and the cyclic prefix duration, respectively. Here we define the \emph{complex equivalent} signal such that the actual (real) transmit signal is given by $\mathbf{s}^{(\cdot)}_{\rm tx}(t) = \mathrm{Re}\{\tilde{\mathbf{s}}^{(\cdot)}_{\rm tx}(t)\}$. For the data symbols, we assume the use of Gaussian signaling with $E^{(\rm d)}_{k} = \mathbb{E}\{{|x_k|}^2\}$, for each $k \in \mathcal{K}$.
The total average transmit OFDM symbol energy (including cyclic prefix) allocated to the UE is defined as $E_{\rm cs}$, where $E_{\rm cs} \geq E^{(\rm r)}$ and $E_{\rm cs} \geq \sum_{k \in \mathcal{K}} E^{(\rm d)}_{k}$.  
For convenience we also assume that $f_{\rm c}$ is a multiple of $1/T_{\rm cs}$, which ensures that the reference tone has the same initial phase in consecutive reference symbols. 

The channel to the representative UE is assumed to be sparse with $L$ resolvable MPCs ($L \ll M_{\rm tx}, M_{\rm rx}$), and the corresponding $M_{\rm rx} \times M_{\rm tx}$ channel impulse response matrix is given as \cite{Akdeniz2014}:
\begin{eqnarray} \label{eqn_channel_impulse_resp}
\mathbf{H}(t) = \sum_{\ell=0}^{L-1} \alpha_{\ell} \mathbf{a}_{\rm rx}(\ell) {\mathbf{a}_{\rm tx}(\ell)}^{\dag} \delta(t - \tau_{\ell}),
\end{eqnarray}
where $\alpha_{\ell}$ is the complex amplitude and $\tau_{\ell}$ is the delay and $\mathbf{a}_{\rm tx}(\ell), \mathbf{a}_{\rm rx}(\ell)$ are the TX and RX array response vectors, respectively, of the $\ell$-th MPC. As an illustration, the $\ell$-th RX array response vector for a uniform planar array with $M_{\rm H}$ horizontal and $M_{\rm V}$ vertical elements ($M_{\rm rx} = M_{\rm H} M_{\rm V}$) is given by $\mathbf{a}_{\rm rx}(\ell) = \bar{\mathbf{a}}_{\rm rx}\big(\psi^{\rm rx}_{\rm azi}(\ell), \psi^{\rm rx}_{\rm ele}(\ell)\big)$, where we define: 
\begin{flalign} \label{eqn_array_response_planar}
& \bar{\mathbf{a}}_{\rm rx}(\psi^{\rm rx}_{\rm azi}, \psi^{\rm rx}_{\rm ele}) \triangleq & \nonumber \\
& \quad \left[\begin{array}{c} 1 \\ e^{{\rm j} 2 \pi \frac{\Delta_{\rm H} \sin(\psi^{\rm rx}_{\rm azi})\sin(\psi^{\rm rx}_{\rm ele})}{\lambda}} \\ \hdots \\ e^{{\rm j} 2 \pi \frac{\Delta_{\rm H} (M_{\rm H}-1) \sin(\psi^{\rm rx}_{\rm azi})\sin(\psi^{\rm rx}_{\rm ele})}{\lambda}}\end{array}\right] \otimes \left[\begin{array}{c} 1 \\ e^{{\rm j} 2 \pi \frac{\Delta_{\rm V} \cos(\psi^{\rm rx}_{\rm ele})}{\lambda}} \\ \hdots \\ e^{{\rm j} 2 \pi \frac{\Delta_{\rm V} (M_{\rm V}-1) \cos(\psi^{\rm rx}_{\rm ele})}{\lambda}}
\end{array}\right], \!\!\!\!\!\!\!&
\end{flalign}
$\psi^{\rm rx}_{\rm azi}(\ell)$, $\psi^{\rm rx}_{\rm ele}(\ell)$ are the azimuth and elevation angles of arrival for the $\ell$-th MPC, $\Delta_{\rm H}, \Delta_{\rm V}$ are the horizontal and vertical antenna spacings and $\lambda$ is the wavelength of the carrier signal. Expressions for $\mathbf{a}_{\rm tx}(\ell)$ can be obtained similarly. 
Note that in \eqref{eqn_channel_impulse_resp} we implicitly assume frequency-flat MPC amplitudes $\{\alpha_{0},..,\alpha_{L-1}\}$ and ignore beam squinting effects \cite{Garakoui2011}, which are reasonable assumptions for moderate system bandwidths. 
To prevent inter symbol interference, we also let the cyclic prefix be longer than the maximum channel delay: $T_{\rm cp} > \tau_{L-1}$. 
To model a time varying channel, we treat $\{\alpha_{\ell}, \mathbf{a}_{\rm tx}(\ell), \mathbf{a}_{\rm rx}(\ell)\}$ as aCSI parameters, that remain constant within an aCSI coherence time and may change arbitrarily afterwards.\footnote{While each MPC may contain several unresolved sub-paths, the corresponding set of scatterers are usually co-located. Therefore the relative sub-path delays and resulting MPC amplitude $\alpha_{\ell}$ are expected to vary slowly with the TX/RX movement.} However since the channel is more sensitive to delay variations, the MPC delays $\{\tau_{0},..., \tau_{L-1}\}$ are modeled as iCSI parameters that only remain constant within a shorter interval called the iCSI coherence time. Note that this time variation of delays is an equivalent representation of the Doppler spread experienced by the RX. 
Finally, we do not assume any distribution prior or side information on $\{\alpha_{\ell}, \mathbf{a}_{\rm tx}(\ell), \mathbf{a}_{\rm rx}(\ell),\tau_{\ell}\}$. 

The RX front-end is assumed to have a low noise amplifier followed by a band-pass filter at each antenna element that leaves the desired signal un-distorted but suppresses the out-of-band noise. The $M_{\rm rx} \times 1$ filtered \emph{complex equivalent} received waveform for the $0$-th symbol can then be expressed as: 
\begin{eqnarray}
\tilde{\mathbf{s}}^{(\cdot)}_{\rm rx}(t) = \sum_{\ell=0}^{L-1} \alpha_{\ell} \mathbf{a}_{\rm rx}(\ell) {\mathbf{a}_{\rm tx}(\ell) }^{\dag} \tilde{\mathbf{s}}^{(\cdot)}_{\rm tx}(t-\tau_{\ell}) + \sqrt{2} \tilde{\mathbf{w}}^{(\cdot)}(t) e^{{\rm j} 2 \pi f_{\rm c} t} \label{eqn_defn_rx_signal}
\end{eqnarray}
for $0 \leq t\leq T_{\rm s}$, where $(\cdot) = ({\rm r})\big/({\rm d})$, $\tilde{\mathbf{w}}^{(\cdot)}(t)$ is the $M_{\rm rx}\times 1$ complex equivalent, base-band, stationary, additive, vector Gaussian noise process, with individual entries being circularly symmetric, independent and identically distributed (i.i.d.), and having a power spectral density: $\mathcal{S}_{\rm w}(f) = \mathrm{N_0}$ for $-f_{K_1} \leq f \leq f_{K_2}$. During the data transmission phase, the $M_{\rm rx} \times 1$ received data waveform $\tilde{\mathbf{s}}^{(\rm d)}_{\rm rx}(t)$ is phase shifted by a bank of phase-shifters, whose outputs are summed and fed to a down-conversion chain for data demodulation, as in conventional analog beamforming. However unlike conventional CE based analog beamforming, the control signals to the phase-shifters are obtained using the reference symbols $\tilde{\mathbf{s}}^{(\rm r)}_{\rm rx}(t)$ and using PACE, as shall be discussed in the next section.

\section{Analog beamformer design at the receiver} \label{sec_analog_phase_amp_est}
During each beamformer design phase, the BS transmits $D$ consecutive reference symbols to facilitate PACE at the RX. This process involves two steps: locking a local RX oscillator to the received reference tone and using this locked oscillator to estimate the amplitude and phase-offsets at each antenna.\footnote{Note that IA based time/frequency synchronization usually involves digital post-processing. Thus prior IA based synchronization does not guarantee that an RX oscillator is locked to the reference tone.} Here locking refers to ensuring that the phase difference between the oscillator and the received reference tone is approximately constant. The first $D_1$ reference symbols are used for the former step and the remaining $D_2=D-D_1$ symbols are used for the latter step. Therefore $D$ is independent of $M_{\rm rx}$ and is mainly determined by the time required for oscillator locking (see Remark \ref{Rem_PLL_parameters}). 
The first step shall be referred to as recovery of the reference tone and is analyzed in Section \ref{subsec_pll_analysis} and while the latter step is discussed in Section \ref{subsec_phase_amp_estimate}. As shall be shown both steps are significantly impaired by channel noise. Therefore in Section \ref{subsec_carrier_arraying}, we propose an improved architecture for reference tone recovery that provides better noise performance, albeit with a slightly higher hardware complexity. For convenience, we shall assume that the MPC delays do not change within the beamformer design phase, and are represented as $\{\hat{\tau}_{0},..., \hat{\tau}_{L-1}\}$ (see also Remark \ref{Rem_MPC_delay_changes}). However the delays may be different during the data transmission phase, as shall be considered in Section \ref{sec_data_transmission}. 
Without loss of generality, assuming the first reference symbol to be the $0$-th OFDM symbol, the complex equivalent RX signal for the $D$ reference symbols at antenna $m$ can be expressed as:\footnote{The component of $\tilde{\mathbf{s}}^{(\cdot)}_{\rm rx}(t)$ for $-T_{\rm cp} \leq t \leq 0$ suffers inter-symbol interference and hence is not included here. \label{note2}}
\begin{eqnarray}
\tilde{s}^{(\rm r)}_{{\rm rx},m}(t) &=& \sqrt{2} A_m^{(\rm r)} e^{{\rm j} 2 \pi f_{\rm c} t} + \sqrt{2} \tilde{w}_{m}^{(\rm r)}(t) e^{{\rm j} 2 \pi f_{\rm c} t} \label{eqn_PLL_input}
\end{eqnarray}
for $0 \leq t \leq D T_{\rm cs} - T_{\rm cp}$, where $A_m^{(\rm r)} \triangleq \sum_{\ell=0}^{L-1} \sqrt{\frac{1}{T_{\rm cs}}} \alpha_{\ell} {[\mathbf{a}_{\rm rx}(\ell)]}_{m} {\mathbf{a}_{\rm tx}(\ell) }^{\dag} \mathbf{t} \sqrt{E^{(\rm r)}} e^{-  {\rm j} 2 \pi f_{\rm c} \hat{\tau}_{\ell}}$ is the amplitude of the reference tone at antenna $m$.

\subsection{Recovery of the reference tone - using one PLL} \label{subsec_pll_analysis}
For locking a local RX oscillator to the reference signal, we first consider the use of a type 2 analog PLL at RX antenna $1$, as illustrated in Fig.~\ref{Fig_PLL_complete}. 
The PLL is a common carrier-recovery circuit - with a mixer, a loop low pass filter ($\mathrm{LF}$) a variable loop gain ($G$) and a voltage controlled oscillator ($\mathrm{VCO}$) arranged in a feedback mechanism - that can filter the noise from an input noisy sinusoidal signal (see \cite{Gupta1975, Viterbi_book} for more details). 
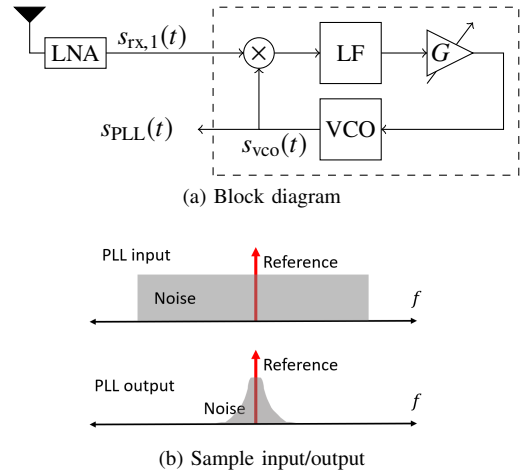
\begin{figure}[h] 
\setlength{\unitlength}{0.08in} 
\centering 
\subfloat[Block diagram]{\label{Fig_PLL_block_diag} \begin{tikzpicture}[x=0.08in,y=0.08in] 
\filldraw[color=black, fill=black](-8,12) (-8,12) -- (-7,13) -- (-9,13) -- cycle;
\draw (-8,12) -- (-8,10) -- (-7,10); \draw(-7,9) rectangle (-3,11) node[pos=.5] {\small LNA}; \draw[->] (-3,10) -- (6,10); \draw(0,11) node{$s_{{\rm rx},1}(t)$}; 
\draw(7,10) circle(1); \draw(7,10) node{$\times$}; \draw[->](8,10) -- (11,10); \draw(11,8) rectangle (15,12) node[pos=.5] {\small LF}; \draw[->] (15,10) -- (18,10); \draw[->] (18,8) -- (21,12); \draw[fill=white] (18,8.5) -- (18,11.5) -- (21,10) -- cycle; \draw(18.9,10) node{$G$}; 
\draw[->] (21,10) -- (23,10) -- (23,5) -- (15,5); 
\draw(11,3) rectangle (15,7) node[pos=.5] {\small VCO}; \draw[->] (11,5) -- (7,5) -- (7,9); 
\draw[->] (7,5)--(3,5); \draw(8,4) node{$s_{\rm vco}(t)$}; \draw[dashed] (4,2) rectangle (24,13); \draw(-1,5) node{$s_{\rm PLL}(t)$};
\end{tikzpicture}} \hspace{15 mm}
\subfloat[Sample input/output]{\includegraphics[width= 0.25\textwidth]{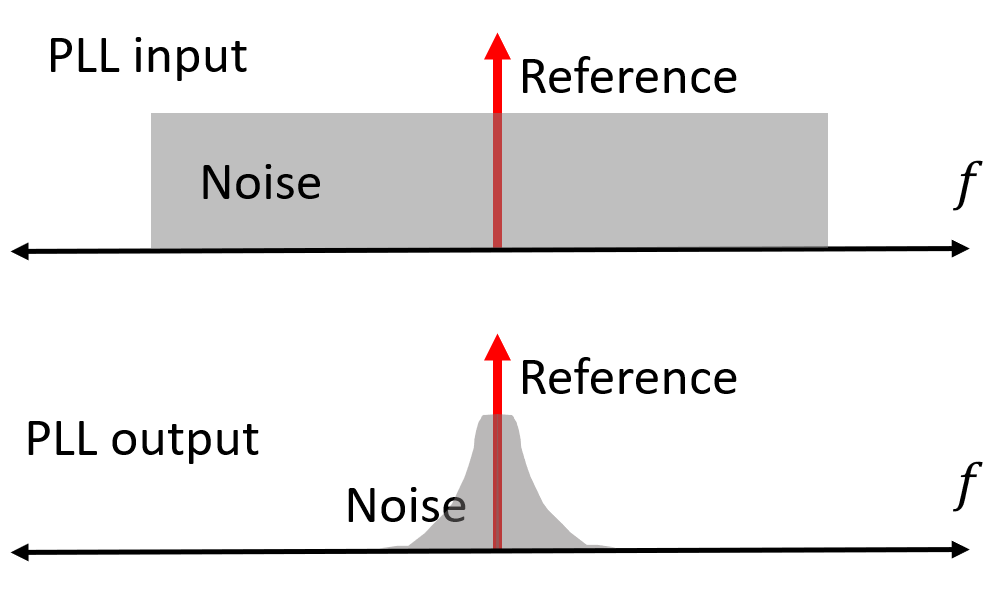} \label{Fig_PLL_illustrate}}
\caption{Block diagram of the PLL at antenna 1 for reference recovery, and a sample illustration of its output. \label{Fig_PLL_complete}} 
\end{figure} 
Here $\mathrm{LF}$ is assumed to be a first-order active low-pass filter with a transfer function $\mathcal{LF}(s) = 1 + \epsilon/s$ and the loop gain $G$ is assumed to adapt to the amplitude of the input such that $G|A_1^{(\rm r)}| = \textrm{constant}$.\footnote{Such a variable gain can possibly be implemented by using an automatic gain control circuit.} 
For convenience, we also ignore the VCO's internal noise \cite{Mehrotra2002, Petrovic2007}. Without loss of generality, let the output of the VCO (i.e. the recovered reference tone) be expressed as: 
\begin{eqnarray}
s_{\rm PLL}(t) = s_{\rm vco}(t) = \sqrt{2} \cos[ 2\pi f_{\rm c} t + \bar{\theta} + \theta(t)] \label{eqn_VCO_output}
\end{eqnarray}
where $\theta(t)$ may be arbitrary and we define $\bar{\theta} \in (-\pi, \pi]$ such that $A_1^{(\rm r)} e^{-{\rm j} \bar{\theta}} = - {\rm j} |A_1^{(\rm r)}|$. Then the stochastic differential equation governing \eqref{eqn_VCO_output} for $0 \leq t \leq D T_{\rm cs} - T_{\rm cp}$ is given by \cite{Viterbi_book}:
\begin{flalign} \label{eqn_PLL_diff_eqn_1}
& 2\pi f_{\rm c} \!+\! \frac{{\rm d} \theta(t)}{{\rm d} t} = \text{LF}\left\{ \mathrm{Re}\{\tilde{s}_{{\rm rx},1}(t)\} \sqrt{2} \cos \big[ 2\pi f_{\rm c} t + \bar{\theta} + \theta(t) \big] \right\} G & \nonumber \\
& \qquad \qquad \qquad \qquad + 2\pi f_{\rm vco} & \nonumber \\
& \quad = \mathrm{LF}\Big\{ \mathrm{Re} \big[A_1^{(\rm r)} e^{-{\rm j} [\bar{\theta} + \theta(t)]} + \tilde{w}^{(\rm r)}_1(t) e^{- {\rm j}[\bar{\theta} + \theta(t)]} \big] \Big\} G + 2\pi f_{\rm vco} \!\!\!\!\!\!\!\! & 
\end{flalign}
where $f_{\rm vco}$ is the free running frequency of the VCO with no input, we use \eqref{eqn_PLL_input} and assume $f_{\rm c}$ is much larger than the bandwidth of $\mathrm{LF}$. In this subsection, we are interested in finding the time required for locking ($D_1 T_{\rm cs}$), i.e., for $\theta(t)$ to (nearly) converge to a constant and characterizing the distribution of the PLL output $s_{\rm PLL}(t)$, or equivalently $\theta(t)$, during the last $D_2$ reference symbols when the PLL is locked to the reference tone. The first part is answered by the following remark:
\begin{remark} \label{Rem_PLL_parameters}
For the PLL considered, the phase lock acquisition time is $\approx \frac{1}{\epsilon} {\left( \frac{2 \pi (f_{\rm c} - f_{\rm vco})}{|A_1^{(\rm r)}| G} \right)}^2$ in the no noise scenario \cite{Gupta1975, Viterbi_book}. Thus $\epsilon$ and $|A_1^{(\rm r)}| G$ must be of the orders of $1/T_{\rm s}$ and $2 \pi |f_{\rm c} - f_{\rm vco}|$ respectively, to keep $D_1$ small. 
\end{remark}
Numerous techniques \cite{Schmitt1979, Ramayya1981} have been proposed to further reduce the lock acquisition time, which are not explored here for brevity. 
In the locked state, it can be shown that $\theta(t)$ suffers from random fluctuations due to the input noise $\tilde{w}^{(\rm r)}_1(t)$ in \eqref{eqn_PLL_diff_eqn_1}, and that $\theta(t)$ (${\rm modulo} \ 2 \pi$) is approximately a zero mean random process \cite{Gupta1975, Viterbi_book}. This fluctuation manifests as phase noise of $s_{\rm PLL}(t)$. 
While several attempts have been made to characterize the locked state $\theta(t)$ (see \cite{Viterbi_book, Gupta1975} and references therein), closed form results are available only for a few simple scenarios that are not applicable here. Therefore, for analytical tractability, 
we linearize \eqref{eqn_PLL_diff_eqn_1} using the following widely used approximations \cite{Viterbi_book}:
\begin{enumerate}
\item We neglect cycle slips and assume that the deviations of $\theta(t)$ about its mean value are small, such that $e^{- {\rm j}\theta(t)} \approx 1 - {\rm j} \theta(t)$ in the locked state.
\item We assume that the distribution of the base-band noise process $\tilde{w}^{(\rm r)}_1(t)$ is invariant to multiplication with $e^{-{\rm j}[\bar{\theta} + \theta(t)]}$, i.e., $\hat{w}^{(\rm r)}_{1}(t) \triangleq \tilde{w}^{(\rm r)}_{1}(t) e^{- {\rm j} [\bar{\theta} + \theta(t)]}$ is also a Gaussian noise process with power spectral density $\mathcal{S}_{\rm w}(f)$.
\end{enumerate}
Approximation 1 is accurate in the locked state and in the large SNR regime, while Approximation 2 is accurate when the noise bandwidth is much larger than the loop filter bandwidth \cite{Viterbi1963, Viterbi_book}. 
Using these approximations and the definition of $\bar{\theta}$, we can linearize \eqref{eqn_PLL_diff_eqn_1} as:
\begin{eqnarray}
\frac{{\rm d} \theta_{\rm L}(t)}{{\rm d} t} &=& \mathrm{LF}\Bigg\{ 
- {|A_1^{(\rm r)}|} \theta_{\rm L}(t) + \frac{\hat{w}^{(\rm r)}_{1}(t)+ {[\hat{w}^{(\rm r)}_{1}(t)]}^{*}}{2} \Bigg\} G \nonumber \\
&& \qquad - 2\pi [f_{\rm c}-f_{\rm vco}] \label{eqn_diff_eqn_reduction}
\end{eqnarray}
where we replace $\theta(t)$ by $\theta_{\rm L}(t)$ to denote use of the linear approximation. Note that for sufficient SNR, $\theta(t) \stackrel{\rm d}{\approx} \theta_{\rm L}(t)$ (${\rm modulo} \ 2 \pi$) during the last $D_2$ reference symbols. Assuming $\theta_{\rm L}(0)=0$ and the PLL input to be $0$ for $t \leq 0$ and taking the Laplace transform on both sides of \eqref{eqn_diff_eqn_reduction}, we obtain: 
\begin{eqnarray}
s \Theta_{\rm L}(s) &=& G \mathcal{LF}(s) \left[ - |A_1^{(\rm r)}| \Theta_{\rm L}(s) + \frac{\hat{\mathcal{W}}^{(\rm r)}_{1}(s) + {[\hat{\mathcal{W}}^{(\rm r)}_{1}(s^{*})]}^{*}}{2} \right] \nonumber \\
&& \qquad - \frac{2\pi [f_{\rm c}-f_{\rm vco}]}{s} \label{eqn_PLL_fourier} 
\end{eqnarray}
where $\Theta_{\rm L}(s)$ and $\hat{\mathcal{W}}^{(\rm r)}_{1}(s)$ are the Laplace transforms of $\theta_{\rm L}(t)$ and $\hat{w}^{(\rm r)}_{1}(t)$, respectively. It can be verified using the final value theorem that the contribution of the last term on the right hand side of \eqref{eqn_PLL_fourier} vanishes for $t \gg 0$ (i.e., in locked state). Therefore ignoring this term in \eqref{eqn_PLL_fourier}, we observe that $\theta_{\rm L}(t)$ is a zero mean, stationary Gaussian process \cite{Mehrotra2002}, in the locked state. Furthermore, the locked state power spectral density, auto-correlation function and variance of $\theta_{\rm L}(t)$ can then be computed, respectively, as: 
\begin{flalign}
& \mathcal{S}_{\theta_{\rm L}}(f) = \mathbb{E} {\left|\Theta_{\rm L}({\rm j}2\pi f) \right|}^2 &\nonumber \\
&\qquad \ \ = \frac{{|G|}^2 (4 \pi^2 f^2 + \epsilon^2) \mathcal{S}_{\rm w}(f) }{ 2{\big| -4 \pi^2 f^2 + G ({\rm j} 2 \pi f + \epsilon) |A_1^{(\rm r)}|\big|}^2} & \label{eqn_PN_spectrum}\\
& \mathcal{R}_{\theta_{\rm L}}(\tau) = \int_{-\infty}^{\infty} \mathcal{S}_{\theta_{\rm L}}(f) e^{{\rm j} 2 \pi f t} {\rm d}t & \nonumber \\
& \qquad \ \  \approx \frac{{|G|}^2 \mathrm{N}_0}{4} \Big[ \frac{a^2 - \epsilon^2}{a(a^2 - b^2)} e^{- a |t|} + \frac{b^2 - \epsilon^2}{b(b^2 - a^2)} e^{- b |t|} \Big] \!\!\!\!\!\!\!\! & \label{eqn_PN_autocorr} \\
& \mathrm{Var}\{\theta_{\rm L}(t)\} = \mathcal{R}_{\theta_{\rm L}}(0) \leq \mathrm{N}_0 \frac{|A_1^{(\rm r)}| G +  \epsilon}{4{|A_1^{(\rm r)}|}^2}, & \label{eqn_PN_var}
\end{flalign}
where $2a = G|A_1^{(\rm r)}| + \sqrt{G^2{|A_1^{(\rm r)}|}^2 - 4 G|A_1^{(\rm r)}| \epsilon}$, $2b = G|A_1^{(\rm r)}| - \sqrt{G^2{|A_1^{(\rm r)}|}^2 - 4 G|A_1^{(\rm r)}| \epsilon}$, \eqref{eqn_PN_autocorr}--\eqref{eqn_PN_var} follow from finding the inverse Fourier transform via partial fraction expansion and the final expressions follow by observing that $\mathcal{S}_{\rm w}(f) \leq \mathrm{N}_0$ for all $f$. Since $\theta_{\rm L}(t)$ is stationary and Gaussian in locked state, note that its distribution is completely characterized by \eqref{eqn_PN_spectrum}--\eqref{eqn_PN_var}. 

\subsection{Phase and amplitude offset estimation} \label{subsec_phase_amp_estimate}
This subsection analyzes the procedure for reference signal phase and amplitude offset estimation at each RX antenna. As illustrated in Fig.~\ref{Fig_block_diag}, the PLL signal from antenna $1$ is fed to a $\pi/2$ phase shifter to obtain its quadrature component. From \eqref{eqn_VCO_output}, the in-phase and quadrature-phase components of the PLL signal for $D_1 T_{\rm cs} - T_{\rm cp} \leq t \leq D T_{\rm cs} - T_{\rm cp}$ can be expressed together as:
\begin{eqnarray}
\tilde{s}_{\rm PLL}(t) = \sqrt{2} e^{{\rm j}[ 2\pi f_{\rm c} t + \bar{\theta} + \theta(t)]}.
\end{eqnarray}
At each RX antenna, the received reference signal is multiplied by the in-phase and quadrature-phase components of the PLL signal, and the resulting outputs are fed to `filter, sample and hold' circuits. This circuit involves a low pass filter with a bandwidth of $\approx 1/(D_2 T_{\rm cs})$, followed by a sample and hold circuit that samples the filtered output at the end of the $D$ reference symbols. 
For convenience, in this paper we shall approximate this `filter, sample and hold' by an integrate and hold operation as depicted in Fig.~\ref{Fig_block_diag}. 
Representing the `filter, sample and hold' outputs corresponding to the in-phase and quadrature-phase components of the PLL output as real and imaginary respectively, the  $M_{\rm rx} \times 1$ complex sample and hold vector can be approximated as:
\begin{flalign}
& \mathbf{I}_{\rm PACE} \approx \frac{1}{D_2} \int_{T_1}^{T_2} \!\!\! \mathrm{Re}\{\tilde{\mathbf{s}}^{(\rm r)}_{\rm rx}(t)\} \tilde{s}^{*}_{\rm PLL}(t) {\rm d}t & \nonumber \\
& \qquad = \frac{1}{D_2} \int_{T_1}^{T_2} \!\! \left[ \sqrt{\frac{1}{T_{\rm cs}}} \boldsymbol{\hat{\mathcal{H}}}(0) \mathbf{t} \sqrt{E^{(\rm r)}} e^{- j [\bar{\theta} + \theta(t)]} \!+\! \hat{\mathbf{w}}^{(\rm r)}(t) \right] {\rm d}t, \!\!\!\!\!\!\!\! & \label{eqn_integral_phase_1}
\end{flalign}
where $\frac{1}{D_2}$ is a scaling factor, $T_1 \triangleq D_1 T_{\rm cs} - T_{\rm cp}$, $T_2 \triangleq D T_{\rm cs} - T_{\rm cp}$, $\boldsymbol{\hat{\mathcal{H}}}(f_k) \triangleq \sum_{\ell=0}^{L-1} \alpha_{\ell} \mathbf{a}_{\rm rx}(\ell) {\mathbf{a}_{\rm tx}(\ell) }^{\dag} e^{- {\rm j} 2 \pi (f_{\rm c}+f_k) \hat{\tau}_{\ell}}$ is the $M_{\rm rx} \times M_{\rm tx}$ frequency-domain channel matrix for the $k$-th  subcarrier during beamformer design phase and $\hat{\mathbf{w}}^{(\rm r)}(t) \triangleq \tilde{\mathbf{w}}^{(\rm r)}(t) e^{- {\rm j} [\bar{\theta} + \theta(t)]}$ is an $M_{\rm rx} \times 1$ i.i.d. Gaussian noise process vector with power spectral density $\mathcal{S}_{\rm w}(f)$ (see Approximation 2). 
Note that in locked state ($T_1 \leq t \leq T_2$), we have $\theta(t) \stackrel{\rm d}{\approx} \theta_{\rm L}(t)$ (${\rm modulo} \ 2 \pi$), as per approximations 1 and 2. Furthermore from \eqref{eqn_PN_autocorr}, the auto-correlation function of $\theta_{\rm L}(t)$ decays exponentially with a time constant of ${\rm O}(1/G |A_1^{(\rm r)}|)$. Therefore, for $G |A_1^{(\rm r)}| \gg 1/ (D_2 T_{\rm cs})$, $\mathbf{I}_{\rm PACE}$ experiences enough independent realizations of $\theta(t)$. Therefore replacing the integral in \eqref{eqn_integral_phase_1} with an expectation over $\mathrm{VCO}$ phase noise, we have:
\begin{eqnarray} \label{eqn_integral_phase_2}
\mathbf{I}_{\rm PACE} & \stackrel{(1)}{\approx} & \sqrt{T_{\rm cs}} \boldsymbol{\hat{\mathcal{H}}}(0) \mathbf{t} \sqrt{E^{(\rm r)}} e^{- {\rm j} \bar{\theta}} \mathbb{E}\{e^{-{\rm j} \theta_{\rm L}(t)}\} + \int_{T_1}^{T_2} \frac{\hat{\mathbf{w}}^{(\rm r)}(t)}{D_2} {\rm d}t \nonumber \\
& \stackrel{(2)}{=}& \sqrt{T_{\rm cs}} \boldsymbol{\hat{\mathcal{H}}}(0) \mathbf{t} \sqrt{E^{(\rm r)}} e^{- {\rm j} \bar{\theta}} e^{-\frac{\mathrm{Var}\{\theta_{\rm L}(t)\}}{2} } + \sqrt{T_{\rm cs}} \hat{\mathbf{W}}^{(\rm r)},
\end{eqnarray}
where $\stackrel{(1)}{\approx}$ follows from the fact that $\theta(t) \stackrel{\rm d}{\approx} \theta_{\rm L}(t)$ (${\rm modulo} \ 2 \pi$) in locked state, $\stackrel{(2)}{=}$ follows by defining $\hat{\mathbf{W}}^{(\rm r)} \triangleq \frac{1}{D_2\sqrt{T_{\rm cs}}} \int_{T_1}^{T_2} \hat{\mathbf{w}}^{(\rm r)}(t) {\rm d}t$ and by using the characteristic function for the stationary Gaussian process $\theta_{\rm L}(t)$. Since $\hat{\mathbf{w}}^{(\rm r)}(t)$ is i.i.d. Gaussian with a power spectral density $\mathcal{S}_{\rm w}(f)$, it can be verified that $\hat{\mathbf{W}}^{(\rm r)} \sim \mathcal{CN}[\mathbb{O}_{M_{\rm rx}\times 1}, (\mathrm{N}_{0}/D_2) \mathbb{I}_{M_{\rm rx}}]$ when $\frac{1}{D_2} \ll K_1,K_2$. From \eqref{eqn_integral_phase_2}, note that the signal component of the sample and hold output $\mathbf{I}_{\rm PACE}$ is directly proportional to the channel matrix at the reference frequency. The outputs are used as a control signals to the RX phase-shifter array, to generate the RX analog beam to be used during the data transmission phase. From \eqref{eqn_integral_phase_2} and \eqref{eqn_PN_var}, note that either $D_2$ or $|A_1^{(\rm r)}|$ can be increased, to reduce the impact of noise $\hat{\mathbf{W}}^{(\rm r)}$ on the analog beam. Since $|A_1^{(\rm r)}|$ is a non-decreasing function of $E^{(\rm r)}$ (see \eqref{eqn_PLL_input}), this implies that $E^{(\rm r)}$ should be kept as large as possible while satisfying $E^{(\rm r)} \leq E_{\rm cs}$ and meeting the spectral mask regulations. 

Note that the results in this section are based on several approximations, including the linear phase noise analysis in Section \ref{subsec_pll_analysis}. To test the accuracy of these results, the numerical values of $\big| \int_{T_1}^{T_2} e^{-{\rm j} \theta(t)}  {\rm d}t \big| \big/ {D_2 T_{\rm cs}}$, obtained by simulating realizations of $\theta(t)$ from \eqref{eqn_PLL_diff_eqn_1}, are compared to its analytic approximation $ e^{- \frac{\mathrm{Var}\{\theta_{\rm L}(t)\}}{2}}$ in Fig.~\ref{Fig_aCE_approx_verify}. Note that this comparison reflects the accuracy of the approximation in \eqref{eqn_integral_phase_2}. As is evident from Fig.~\ref{Fig_aCE_approx_verify}, \eqref{eqn_integral_phase_2} is accurate above a certain SNR. Additionally, since $\mathbf{I}_{\rm PACE}$ decays exponentially with $\mathrm{Var}\{\theta_{\rm L}(t)\}$ (see \eqref{eqn_integral_phase_2}), we observe from Fig.~\ref{Fig_aCE_approx_verify_mean} that the mean integrator output drops drastically below a certain threshold SNR. As shall be shown in Section \ref{sec_data_transmission}, such a drop in the mean causes a sharp degradation in the system performance below this threshold SNR. 
Therefore in the next subsection we propose a better reference recovery circuit, called weighted carrier arraying, that reduces the SNR threshold. 
\begin{figure}[h] 
\centering 
\subfloat[$1-$Mean: $1- \mathbb{E}\big| \int_{T_1}^{T_2} \frac{e^{-{\rm j} \theta(t)}}{D_2 T_{\rm cs}} {\rm d}t \big|$]{\includegraphics[width= 0.45\textwidth]{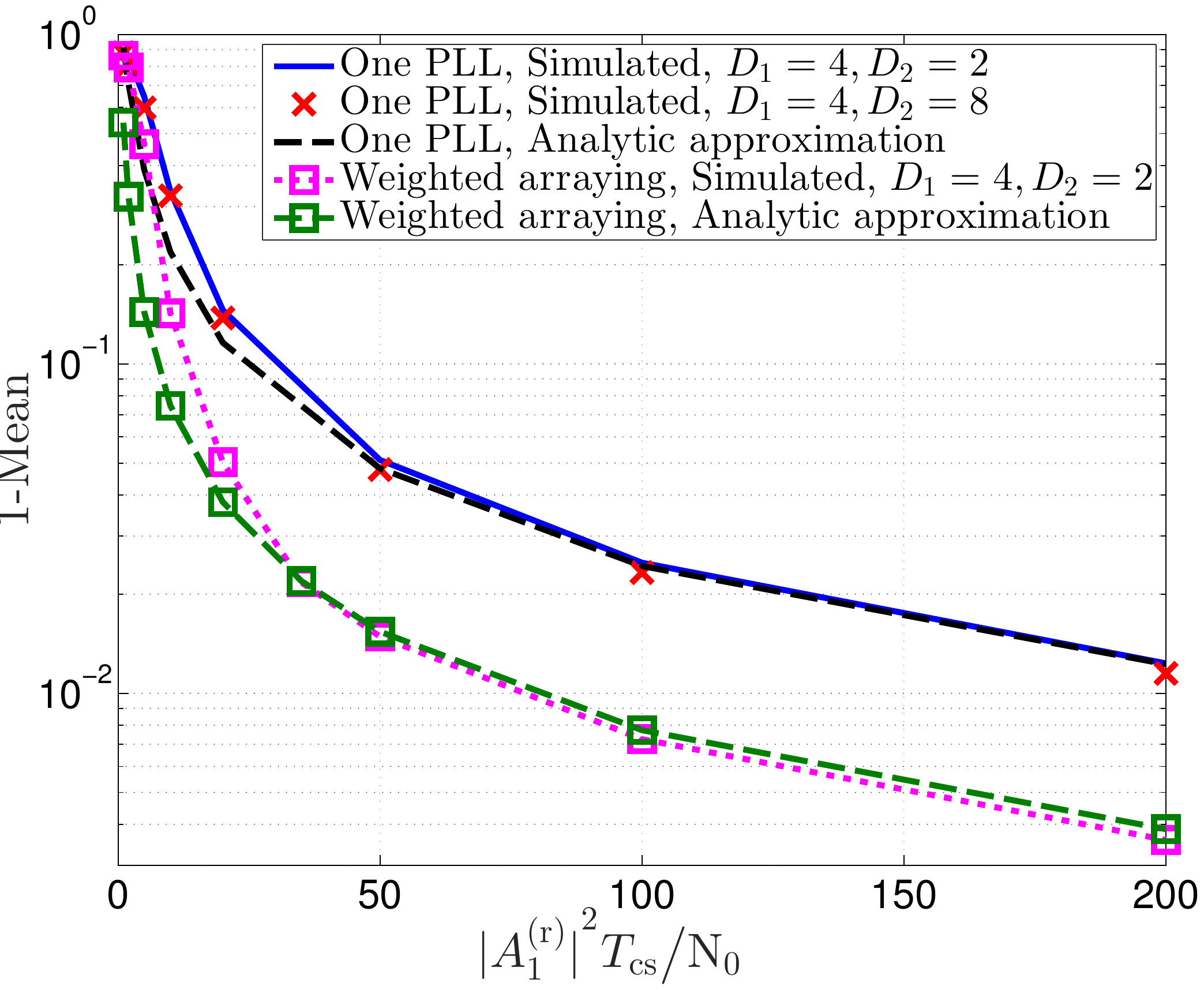} \label{Fig_aCE_approx_verify_mean}} \hspace{0.5cm}
\subfloat[Variance: $\mathbb{E} {\Big[ \Big| \int_{T_1}^{T_2} \!\! \frac{e^{-{\rm j} \theta(t)} }{D_2 T_{\rm cs}} {\rm d}t \Big| - e^{-\frac{\mathrm{Var}\{\theta_{\rm L}\}}{2}} \Big]}^2$]{\includegraphics[width= 0.45\textwidth]{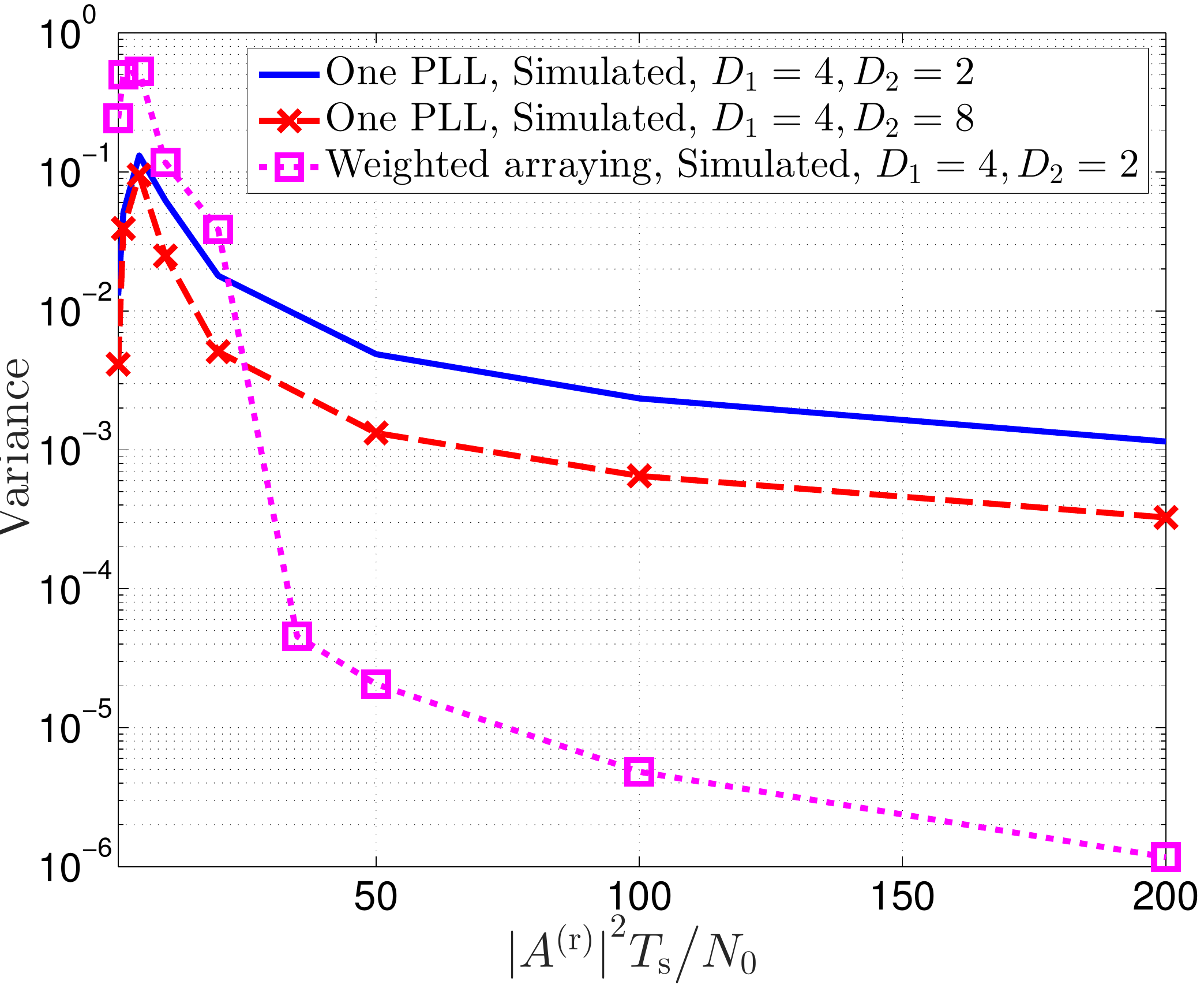} \label{Fig_aCE_approx_verify_var}}
\caption{Accuracy of the analytical approximation for the filter, sample and hold outputs in \eqref{eqn_integral_phase_2} versus SNR for the one PLL and weighted arraying architectures. In Fig.~\ref{Fig_aCE_approx_verify_mean}, we plot $1- \mathbb{E}\big| \int_{T_1}^{T_2} \frac{e^{-{\rm j} \theta(t)}}{D_2 T_{\rm cs}} {\rm d}t \big|$ for simulations and $1- e^{-\frac{\mathrm{Var}\{\theta_{\rm L}\}}{2}}$ for analytic approximation. We assume $A^{(\rm r)}_1 = 1$, $A^{(\rm r)}_5 = 0.7e^{{\rm j} \pi/3}$, $A^{(\rm r)}_{15} = 0.5 e^{-{\rm j} \pi/3}$ and the remaining parameters are from Table \ref{Table_sim_param}.} 
\label{Fig_aCE_approx_verify} 
\end{figure} 
\begin{table}[ht]
\caption{One PLL and weighted arraying simulation parameters} 
\centering  
\begin{tabular}{| l | c | l | c |} 
\hline
Parameter & Value & Parameter & Value \\
\hline
$f_{\rm c}$ & $30$GHz & $\epsilon$ & $4/T_{\rm s}$ \\
\hline
$f_{\rm c} - f_{\rm vco}$ & $5$MHz & $f_{\rm IF}$ & $1$GHz \\
\hline
$T_{\rm s}$ & $1 \mu s$ & $f_{\rm c} \!-\! f_{\rm IF} \!-\! f^{\rm p}_{\rm vco}$ & $5$MHz \\
\hline
$T_{\rm cp}$ & $0.1 \mu$s & $\mathcal{M}$ & $\{1,5,15\}$ \\
\hline
$K_1$ & $512$ & $\mu$ & $2\pi/T_{\rm s}$ \\
\hline
$K_2$ & $511$ & $G^{\rm p} {|A_{\rm rss}^{(\rm r)}|}^2 \big/ \mu $ & $\pi |f_{\rm c} \!-\! f_{\rm IF} \!-\! f^{\rm p}_{\rm vco}|$ \\
\hline
$G |A_1^{(\rm r)}|$ & $\pi|f_{\rm c} \!-\! f_{\rm vco}|$ & $\epsilon^{\rm p}$ & $4/T_{\rm s}$ \\
\hline
\end{tabular}
\label{Table_sim_param} 
\end{table}
%

\begin{remark} \label{Rem_MPC_delay_changes}
The preceding derivations assumed that the MPC delays are identical for the $D$ reference symbols. However since the PLL continuously tracks the RX signal and phase/amplitude estimation at each antenna is performed simultaneously, these results are valid even if the delays change slowly within the beamformer design phase.
\end{remark}
\begin{remark} \label{Rem_parallel_data_tx}
The RX phase-shifter array or the down-conversion chain are not utilized during the $D$ reference symbols of the beamformer design phase. Therefore, data reception is also possible during these $D$ reference symbols in parallel, as long as a sufficient guard band between the data sub-carriers and the reference sub-carrier is provided (similar to \eqref{eqn_tx_signal_pilot}) to reduce impact on the PLL performance. 
\end{remark}
Note that in a multi-cell scenario, use of the same reference tone in adjacent cells can cause reference tone contamination, i.e., $\mathbf{I}_{\rm PACE}$ may contain components corresponding to the channel from a neighboring BS. This is analogous to pilot contamination in conventional CE approaches \cite{Marzetta}, and can be avoided by using different, well-separated reference frequencies in adjacent cells. 

\subsection{Recovery of the reference tone - using weighted carrier arraying} \label{subsec_carrier_arraying}
For reducing the PLL SNR threshold and improving performance, in this subsection we propose a new reference recovery technique called weighted carrier arraying, as illustrated in Fig.~\ref{Fig_block_diag_arraying}.
Apart from a main primarly PLL, weighted carrier arraying has secondary PLLs at a subset $\mathcal{M}$ of antennas, which compensate for the inter-antenna phase shift. The resulting phase compensated signals from the $\mathcal{M}$ antennas are weighted, combined and tracked by the primary PLL, which operates at a higher SNR and with a wider loop bandwidth than the secondary PLLs. Note that this architecture can be interpreted as a generalization of the carrier recovery process in \cite{Schrader1962, Schrader1964, Thompson1976, Golliday1982} that allows weighted combining. We shall next analyze the performance of this arrayed PLL in the locked state. However, an analysis of the transient behavior and lock acquisition time of this design is beyond the scope of this paper.
\begin{figure*}[h] 
\setlength{\unitlength}{0.08in} 
\centering 
\begin{tikzpicture}[x=0.08in,y=0.08in] 
\draw[pattern=crosshatch dots, pattern color=gray] (21.5,0) rectangle (34.5,17); \draw(28,8.5) node{\small Secondary PLLs};
\draw (12.5,2) -- (12.5,19);
\filldraw[color=black, fill=black](0,7) (0,7) -- (1,8) -- (-1,8) -- cycle;
\draw (0,7) -- (0,5) -- (1,5); \draw(1,4) rectangle (5,6) node[pos=.5] {\small LNA}; \draw[->] (5,5) -- (14,5); \draw(8,6) node{$s_{{\rm rx},m}(t)$}; \draw (15,5) circle(1); \draw(15,5) node{$\times$}; \filldraw[fill=black] (12.5,2) circle(0.3); \draw[->] (12.5,2) -- (15,2) -- (15,4); \draw(16,5) -- (17,5); \draw(17,4) rectangle (21,6) node[pos=.5] {\small ${\rm LPF}$}; \draw[->] (21,5)--(22,5);
\draw[fill=white] (23,5) circle(1); \draw(23,5) node{$\times$}; \draw(24,5) -- (25,5); \draw[fill=white] (25,4) rectangle (29,6) node[pos=.5] {\small ${\rm LPF}$}; \draw[->] (29,5) -- (34,5) -- (34,2) -- (32,2); \draw[->] (30.5,0.5) -- (33.5,4); 
\draw[fill=white] (33,0.5) -- (33,3.5) -- (30,2) -- cycle; \draw(31.9,2) node{\small $G^{\rm s}_m$}; \draw[->] (30,2)--(29,2);
\draw[fill=white] (24,1) rectangle (29,3) node[pos=.5] {\small ${\rm VCO}_{\rm s}$}; \draw[->] (24,2) -- (23,2) -- (23,4); 
\draw (34,5) -- (35.5,5); \draw[->] (35.5,3) -- (38.5,6.5); \draw[fill=white] (35.5,3.5) -- (35.5,6.5) -- (38.5,5) -- cycle; \draw(38.5,3) node{\small $1/G^{\rm s}_m$}; \draw[->] (38.5,5) -- (40,5) -- (40,9);

\filldraw[color=black, fill=black](0,17) (0,17) -- (1,18) -- (-1,18) -- cycle;
\draw (0,17) -- (0,15) -- (1,15); \draw(1,14) rectangle (5,16) node[pos=.5] {\small LNA}; \draw[->] (5,15) -- (14,15); \draw(8,16) node{$s_{{\rm rx},1}(t)$}; \draw (15,15) circle(1); \draw(15,15) node{$\times$}; \filldraw[fill=black] (12.5,12) circle(0.3); \draw[->] (12.5,12) -- (15,12) -- (15,14); \draw(16,15) -- (17,15); \draw(17,14) rectangle (21,16) node[pos=.5] {\small ${\rm LPF}$}; \draw[->] (21,15)--(22,15);
\draw[fill=white] (23,15) circle(1); \draw(23,15) node{$\times$}; \draw(24,15) -- (25,15); \draw[fill=white] (25,14) rectangle (29,16) node[pos=.5] {\small ${\rm LPF}$}; \draw[->] (29,15) -- (34,15) -- (34,12) -- (32,12); \draw[->] (30.5,10.5) -- (33.5,14); 
\draw[fill=white] (33,10.5) -- (33,13.5) -- (30,12) -- cycle; \draw(31.9,12) node{\small $G^{\rm s}_1$}; \draw[->] (30,12)--(29,12);
\draw[fill=white] (24,11) rectangle (29,13) node[pos=.5] {\small ${\rm VCO}_{\rm s}$}; \draw[->] (24,12) -- (23,12) -- (23,14); 
\draw (34,15) -- (35.5,15); \draw[->] (35.5,13) -- (38,16.5); \draw[fill=white] (35.5,13.5) -- (35.5,16.5) -- (38.5,15) -- cycle; \draw(38.5,17.5) node{\small $1/G^{\rm s}_1$}; \draw[->] (38.5,15) -- (40,15) -- (40,11);

\draw (40,10) circle(1); \draw(40,10) node{$+$}; \draw[->] (41,10) -- (42,10); \draw(42,9) rectangle (45,11) node[pos=.5] {\small ${\rm LF}_{\rm p}$}; \draw (45,10) -- (46,10); \draw[->] (46,8) -- (49,12); \draw[fill=white] (46,8.5) -- (46,11.5) -- (49,10) -- cycle; \draw(46.9,10) node{\small ${G}^{\rm p}$}; \draw[->] (49,10) -- (50,10) -- (50,12.5); \draw (50,14.5) -- (50,19) -- (12.5,19);
\draw(47,12.5) rectangle (53,14.5) node[pos=.5] {\small ${\rm VCO}_{\rm p}$}; 
\draw[->] (50,19) -- (51,19); \draw (52,19) circle(1); \draw(52,19) node{$\times$}; \draw[->] (52,17) -- (52,18); \draw(56,16) node{\small $\cos(2\pi f_{\rm IF}t)$};
\draw (53,19) -- (54,19); \draw(54,18) rectangle (58,20) node[pos=.5] {\small ${\rm BPF}$}; \draw[->] (58,19) -- (60,19); \draw(63,19) node{$s_{\rm PLL}(t)$}; \draw[dashed] (11.5,-1) rectangle (59,21); \draw(53,0) node{Arrayed PLL};
\end{tikzpicture}
\caption{Block diagram of weighted carrier arraying for reference tone recovery.} 
\label{Fig_block_diag_arraying} 
\end{figure*}
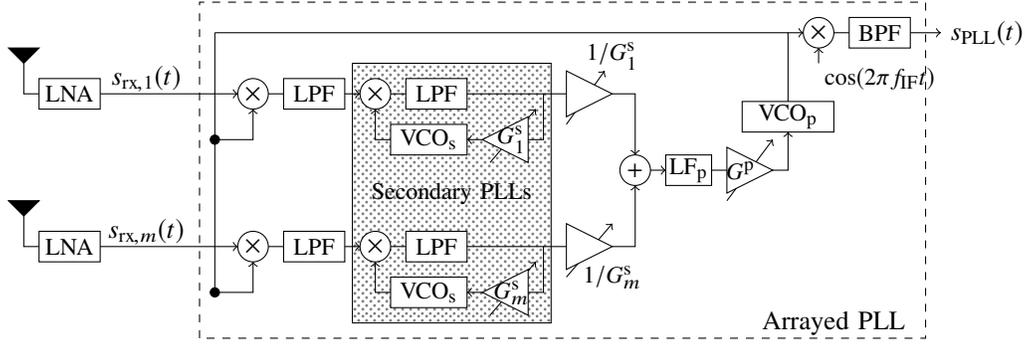 
In Fig.~\ref{Fig_block_diag_arraying}, $\mathrm{LPF}/\mathrm{BPF}$ refer to low-pass and band-pass filters with wide bandwidths, designed only to remove the unwanted side-band of the mixer outputs. Without loss of generality, we express the outputs of the primary and secondary VCOs as:\footnote{Another convergence point for $s^{\rm p}_{\rm vco}(t)$ is at a frequency of $(f_{\rm c}+f_{\rm IF})$. But the final results presented here are also valid for this alternate convergence point.}
\begin{eqnarray}
s^{\rm p}_{\rm vco}(t) &=& \sqrt{2} \cos[ 2\pi (f_{\rm c}-f_{\rm IF}) t + \theta(t)] \nonumber \\
s^{\rm s}_{\rm vco,m}(t) &=& \sqrt{2} \cos[ 2\pi f_{\rm IF} t + \bar{\phi}_m + \phi_m(t)], \ \ m \in \mathcal{M} \nonumber
\end{eqnarray}
respectively, where $\theta(t),\phi_m(t)$ are arbitrary, $f_{\rm IF}$ is the common free running frequency of the secondary VCOs, and $\bar{\phi}_m$ are such that $A_m^{(\rm r)} e^{-{\rm j} [\bar{\phi}_m]} = -{\rm j} |A_m^{(\rm r)}|$ for all $m \in \mathcal{M}$. Now similar to Section \ref{subsec_pll_analysis}, from \eqref{eqn_PLL_input} the differential equation governing the secondary PLL at antenna $m \in \mathcal{M}$ can be expressed as:
\begin{eqnarray}
\frac{{\rm d} \phi_m(t)}{{\rm d} t} &=& \mathrm{Re} \big[A_m^{(\rm r)} e^{-{\rm j} [\bar{\phi}_m + \phi_m(t) + \theta(t)]} \nonumber \\
&& \qquad + \tilde{w}_m^{(\rm r)}(t) e^{- {\rm j}[\bar{\phi}_m + \phi_m(t) + \theta(t)]} \big] \frac{G^{\rm s}_m}{\sqrt{2}} \nonumber \\
&=& \mathrm{Re} \big[- {\rm j}|A_m^{(\rm r)}| e^{-{\rm j} [\phi_m(t) + \theta(t)]} + \hat{w}_m^{(\rm r)}(t) \big] \frac{{G}^{\rm s}_m}{\sqrt{2}}  \label{eqn_carr_array_sec}
\end{eqnarray}
where we define $\hat{w}_m^{(\rm r)}(t) \triangleq \tilde{w}_m^{(\rm r)}(t) e^{- {\rm j}[\bar{\phi}_m + \phi_m(t) + \theta(t)]}$ and $G^{\rm s}_m$ is the loop gain of the secondary VCO at antenna $m$. Similarly, for the primary VCO we have:
\begin{eqnarray}
2\pi (f_{\rm c}-f_{\rm IF}) + \frac{{\rm d} \theta(t)}{{\rm d} t} &=& \mathrm{LF} \Big\{ \sum_{m \in \mathcal{M}} \mathrm{Re} \big[- {\rm j}|A_m^{(\rm r)}| e^{-{\rm j} [\phi_m(t) + \theta(t)]} \nonumber \\
&& \quad + \hat{w}_m^{(\rm r)}(t) \big] \frac{1}{G_m^{\rm s}} \Big\} \frac{{G}^{\rm p}}{\sqrt{2}} + 2\pi f^{\rm p}_{\rm vco} \label{eqn_carr_array_prim}
\end{eqnarray}
where $f^{\rm p}_{\rm vco}$ is the free running frequency of the primary VCO, $G^{\rm p}$ is the loop gain and $\mathrm{LF}_{\rm p}$ is an active low pass filter with transfer function $\mathrm{LF}_{\rm p}(s) = (1+\epsilon^{\rm p}/s)$. Similar to Section \ref{subsec_pll_analysis}, to obtain the locked state distribution of $\theta(t)$ we shall rely on the linear PLL analysis by using: 1) $e^{-{\rm j} [\phi_m(t) + \theta(t)]} \approx 1 -{\rm j} [\phi_m(t) + \theta(t)]$, which is accurate in the high SNR locked state where $\phi_m(t) + \theta(t) \ll 1$ and 2) $\hat{w}_m^{(\rm r)}(t) \stackrel{\rm d}{\approx} \tilde{w}_m^{(\rm r)}(t)$, which is accurate for a wide noise bandwidth. Using these approximations in \eqref{eqn_carr_array_sec}--\eqref{eqn_carr_array_prim} with zero initial conditions and taking Laplace transforms, we obtain:
\begin{subequations} \label{eqn_array_laplace}
\begin{flalign}
&s\Phi^{\rm L}_m(s)  = \Big( -|A_m^{(\rm r)}| [\Phi^{\rm L}_m(s) + \Theta_{\rm L}(s)] & \nonumber \\
& \qquad \qquad \qquad \qquad + \frac{\hat{\mathcal{W}}_m^{(\rm r)}(s) + {[\hat{\mathcal{W}}_m^{(\rm r)}(s^{*})]}^{*}}{2} \Big) \frac{G^{\rm s}_m}{\sqrt{2}} & \\
& s\Theta_{\rm L}(s) = \mathrm{LF}(s) \!\! \sum_{m \in \mathcal{M}} \!\! \Big[ \!\! -\! \frac{|A_m^{(\rm r)}|}{G_{m}^{\rm s}} [\Phi^{\rm L}_m(s) \!+\! \Theta_{\rm L}(s)] & \nonumber \\
& \qquad \quad \!+\! \frac{\hat{\mathcal{W}}_m^{(\rm r)}(s) \!+\! {[\hat{\mathcal{W}}_m^{(\rm r)}(s^{*})]}^{*}}{2 G_{m}^{\rm s}} \Big] \frac{{G}^{\rm p}}{\sqrt{2}} \!+\! \frac{2\pi (f_{\rm IF}\!+\! f^{\rm p}_{\rm vco}\!-\!f_{\rm c})}{s} \!\!\!\!\!\!\!\!\!\! & 
\end{flalign}
\end{subequations}
where $\hat{\mathcal{W}}^{(\rm r)}_{m}(s)$, $\Theta_{\rm L}(s)$ and $\Phi^{\rm L}_m(s)$ are the Laplace transforms of $\hat{w}^{(\rm r)}_{m}(t)$, linear approximation $\theta_{\rm L}(t)$ and linear approximation $\phi^{_{\rm L}}_m(t)$, respectively. We assume that the loop gains of the PLLs adapt to the amplitudes of the input such that $|A_m^{(\rm r)}| G^{\rm s}_m = \mu \ \forall m \in \mathcal{M}$ and $\sum_{m \in \mathcal{M}} G^{\rm p} {|A_m^{(\rm r)}|}^2 = \text{constant}$. Then solving the system of equations in \eqref{eqn_array_laplace}, we obtain:
\begin{flalign} \label{eqn_laplace_arraying_final}
&\bigg[s + \sum_{m \in \mathcal{M}} \frac{(s + \epsilon^{\rm p}){|A_m^{(\rm r)}|}^2 G^{\rm p}}{\mu(\sqrt{2} s + \mu)} \bigg] \Theta_{\rm L}(s) & \nonumber \\
& \qquad \qquad = \sum_{m \in \mathcal{M}} \frac{(s + \epsilon^{\rm p}) (\hat{\mathcal{W}}_m^{(\rm r)}(s) + {[\hat{\mathcal{W}}_m^{(\rm r)}(s^{*})]}^{*})|A_m^{(\rm r)}| {G}^{\rm p}}{2 \mu (\sqrt{2} s + \mu)} & \nonumber \\
& \qquad \qquad \qquad \qquad + \frac{2\pi (f_{\rm IF} + f^{\rm p}_{\rm vco} - f_{\rm c})}{s} &
\end{flalign}
It can be verified using the final value theorem that the last term in \eqref{eqn_laplace_arraying_final} only contributes a constant phase shift for $t \gg 0$ (in locked state), say $\bar{\theta}_{\rm L}$.\footnote{Simulations suggest this constant phase shift for the actual non-linear system \eqref{eqn_carr_array_sec}--\eqref{eqn_carr_array_prim} is noise dependent. However such an arbitrary, but constant, phase shift does not impact the resulting beamforming gain if cycle skipping probability is low.}
Thus, using steps similar to Section \ref{subsec_pll_analysis}, we can obtain the locked state power spectral density and variance of the time varying part of $\theta_{\rm L}(t)$, i.e., $\theta_{\rm L}(t)-\bar{\theta}_{\rm L}$, as:
\begin{flalign} 
&\mathcal{S}_{\theta_{\rm L}-\bar{\theta}_{\rm L}}(f) & \nonumber \\
& \qquad = \frac{\mathrm{N}_0 {|A^{(\rm r)}_{\rm rss}{G}^{\rm p} |}^2}{2} { \left| \frac{ (s + \epsilon^{\rm p}) }{ s \mu(\sqrt{2} s + \mu) + (s + \epsilon^{\rm p}) {|A^{(\rm r)}_{\rm rss}|}^2 G^{\rm p}} \right|}_{s = {\rm j 2 \pi f}}^2 \!\!\!\!\!\!\!\!\!\!\!\!\!\!\! & \label{eqn_PN_spectrum_array} \\
& \mathrm{Var}\{\theta_{\rm L}(t)\} = \frac{({|A^{(\rm r)}_{\rm rss}|}^2 [G^{\rm p}/\mu ]+ \sqrt{2} \epsilon^{\rm p}) [{G}^{\rm p}/\mu] \mathrm{N}_0}{4 \sqrt{2}(\mu + {|A^{(\rm r)}_{\rm rss}|}^2 [G^{\rm p}/\mu])} & \nonumber \\
& \qquad \qquad \ \leq \frac{({|A^{(\rm r)}_{\rm rss}|}^2 [G^{\rm p}/\sqrt{2} \mu ]+ \epsilon^{\rm p}) \mathrm{N}_0}{4 {|A^{(\rm r)}_{\rm rss}|}^2}, & \label{eqn_PN_var_array}
\end{flalign}
where ${[A^{(\rm r)}_{\rm rss}]}^2 = \sum_{m \in \mathcal{M}} {|A_m^{(\rm r)}|}^2$. 
Comparing \eqref{eqn_PN_var_array} to \eqref{eqn_PN_var}, note that the PLL phase noise is essentially reduced by the maximal ratio combining gain corresponding to the $\mathcal{M}$ antennas. As this variation in $\theta_{\rm L}(t)$ manifests as phase noise of $s_{\rm PLL}(t)$ in Fig.~\ref{Fig_block_diag_arraying}, the `filter, sample and hold' outputs with weighted carrier arraying can be obtained by using \eqref{eqn_PN_var_array} in \eqref{eqn_integral_phase_2}. The accuracy of the resulting approximation is studied via simulations in Fig.~\ref{Fig_aCE_approx_verify}. 

\section{Data transmission} \label{sec_data_transmission}
During the data transmission phase, OFDM symbols of type \eqref{eqn_tx_signal_data} are transmitted and the corresponding received signals are processed via the phase-shifter array with $\mathbf{I}_{\rm PACE}$ as the control signals. Without loss of generality, again assuming the $0$-th OFDM symbol as a representative data symbol, the combined data signal at the RX for $0 \leq t \leq T_{\rm s}$ can be expressed as:\footref{note2}
\begin{flalign}
& R(t) = \frac{1}{\sqrt{2}} \mathbf{I}_{\rm PACE}^{\dag} \bigg[ \sum_{\ell=0}^{L-1} \sum_{k \in \mathcal{K}} \sqrt{\frac{2}{T_{\rm cs}}} \alpha_{\ell} \mathbf{a}_{\rm rx}(\ell) {\mathbf{a}_{\rm tx}(\ell) }^{\dag} \mathbf{t} x_k e^{{\rm j} 2 \pi (f_{\rm c}+f_k) (t-\tau_{\ell})} & \nonumber \\
& \qquad \qquad \qquad + \sqrt{2}\tilde{\mathbf{w}}^{(\rm d)}(t) e^{{\rm j} 2 \pi f_{\rm c} t} \bigg] & \nonumber 
\end{flalign}
where the $1/\sqrt{2}$ is a scaling constant for convenience and we assume that the MPC delays for this representative data symbol are $\{\tau_0,...,\tau_{L-1}\}$. 
This phase shifted and combined signal $R(t)$ is then converted to base-band by a separate RX oscillator, and any resulting phase noise is assumed to be mitigated via some digital phase noise compensation techniques \cite{Robertson1995, Wu2006, Petrovic2007, Randel2010}. 
Therefore neglecting the down-conversion phase noise, the resulting base-band signal can be expressed as $R_{\rm BB}(t) = R(t) e^{-{\rm j} 2 \pi f_{\rm c} t}$. 
This signal is then sampled and OFDM demodulation follows. The OFDM demodulation output for the $k$-th subcarrier ($k \in \mathcal{K}$) is then given by:
\begin{eqnarray}
Y_{k} &=& \frac{1}{K}\sum_{u = 0}^{K-1} R_{\rm BB} \Big(\frac{u T_{\rm s}}{K}\Big) e^{- {\rm j} \frac{2 \pi k u}{K}} \nonumber \\
&=& \frac{1}{\sqrt{T_{\rm cs}}} \mathbf{I}_{\rm PACE}^{\dag} \boldsymbol{\mathcal{H}}(f_k) \mathbf{t}  x_k + \frac{1}{\sqrt{T_{\rm cs}}}\mathbf{I}_{\rm PACE}^{\dag} \tilde{\mathbf{W}}^{(\rm d)}[k] \label{eqn_exprn_Y_k}
\end{eqnarray}
where $\boldsymbol{\mathcal{H}}(f_k) \triangleq \sum_{\ell=0}^{L-1} \alpha_{\ell} \mathbf{a}_{\rm rx}(\ell) {\mathbf{a}_{\rm tx}(\ell) }^{\dag} e^{- {\rm j} 2 \pi (f_{\rm c}+f_k) \tau_{\ell}}$ is the $M_{\rm rx} \times M_{\rm tx}$ frequency domain channel matrix for the $k$-th data subcarrier and $\tilde{\mathbf{W}}^{(\rm d)}[k] \triangleq \frac{\sqrt{T_{\rm cs}}}{K} \sum_{u=0}^{K-1} \tilde{\mathbf{w}}^{(\rm d)}(\frac{u T_{\rm s}}{K}) e^{- {\rm j} \frac{2 \pi k u}{K}}$, with $\tilde{\mathbf{W}}^{(\rm d)}[k]$ being independently distributed for each $k \in \mathcal{K}$ as $\tilde{\mathbf{W}}^{(\rm d)}[k] \sim \mathcal{CN}[\mathbb{O}_{M_{\rm rx} \times 1}, (\mathrm{N}_0 T_{\rm cs}/T_{\rm s}) \mathbb{I}_{M_{\rm rx}} ]$. 
Note from \eqref{eqn_integral_phase_2} that $\mathbf{I}^{\dag}_{\rm PACE}$ is similar (with appropriate scaling), but not identical, to the MRC beamformer for the $k$-th sub-carrier: $\mathbf{t}^{\dag} \boldsymbol{\mathcal{H}}(f_k)^{\dag}$. The mismatch is due to the beamforming noise $\hat{\mathbf{W}}^{(\rm r)}$ and because the reference symbols and the $k$-th sub-carrier data stream pass through slightly different channels, owing to the difference in sub-carrier frequencies and the MPC delays ($\hat{\tau}_{\ell} \neq \tau_{\ell}$). Consequently, the beamformer $\mathbf{I}_{\rm PACE}$ only achieves imperfect MRC, leading to some loss in performance and causing the effective channel coefficients $\mathbf{I}_{\rm PACE}^{\dag} \boldsymbol{\mathcal{H}}(f_k) \mathbf{t}$ to vary with the sub-carrier index $k$, i.e., the system experiences frequency-selective fading. Furthermore, since the MPC delays $\{\tau_{0},..,\tau_{L-1}\}$ change after every iCSI coherence time, so may these channel coefficients. 
As depicted in Fig.~\ref{Fig_tx_protocol}, we assume that the TX transmits pilot symbols within each iCSI coherence time to facilitate estimation of these coefficients $\big\{ \mathbf{I}_{\rm PACE}^{\dag} \boldsymbol{\mathcal{H}}(f_k) \mathbf{t} \big| k \in \mathcal{K}\big\}$ at the RX. Since these pilots are used only to estimate the effective single-input-single-output (SISO) channel and not the actual MIMO channel, the corresponding overhead is small and shall be neglected here. Assuming perfect estimates of these channel coefficients, from \eqref{eqn_exprn_Y_k} the effective SNR for the $k$-th sub-carrier, and the instantaneous system spectral efficiency (iSE), respectively, can be expressed as:
\begin{flalign} 
& \gamma^{\rm PACE}_{k}(\hat{\mathbf{W}}^{(\rm r)}, \mathbf{H}(t)) \triangleq \frac{{|\mathbf{I}_{\rm PACE}^{\dag} \boldsymbol{\mathcal{H}}(f_k) \mathbf{t}|}^2 E_k^{(\rm d)}}{ {\|\mathbf{I}_{\rm PACE}\|}^2 \mathrm{N}_0 T_{\rm cs}/T_{\rm s} } & \label{eqn_SNR_exprn} \\
& {\rm iSE}^{\rm PACE} \big(\hat{\mathbf{W}}^{(\rm r)}, \mathbf{H}(t) \big) \triangleq \sum_{k \in \mathcal{K}} \frac{1}{K} \log \big(1 + \gamma^{\rm PACE}_{k}(\hat{\mathbf{W}}^{(\rm r)}, \mathbf{H}(t)) \big), \!\!\!\!\!\!\!\! & \label{eqn_cap}
\end{flalign}
where we neglect the cyclic prefix overhead in \eqref{eqn_cap} for convenience. Note that the iSE maximizing data power allocation $\{ E_k^{(\rm d)} | k \in \mathcal{K}\}$ can be obtained via water-filling across the sub-carriers. While the exact expressions for \eqref{eqn_SNR_exprn}--\eqref{eqn_cap} are involved, their expectations with respect to $\mathbf{I}_{\rm PACE}$ can be bounded, as stated by the following theorem.
\begin{theorem} \label{Theorem1}
If the RX array response vectors for the channel MPCs are mutually orthogonal, i.e., $\mathbf{a}_{\rm rx}(\ell)^{\dag} \mathbf{a}_{\rm rx}(i) = 0$ for $\ell \neq i$, the effective SNR and iSE, averaged over the beamformer noise $\hat{\mathbf{W}}^{(\rm r)}$, can be bounded as in \eqref{eqn_lemma_perf_anal}
\begin{figure*}[!h]
\begin{subequations} \label{eqn_lemma_perf_anal}
\begin{eqnarray} 
\gamma^{\rm PACE}_{k}(\mathbf{H}(t)) & \gtrapprox & \frac{ M_{\rm rx} E^{(\rm r)} e^{-\mathrm{Var}\{\theta_{\rm L}(t)\}} {|\beta(f_{\rm c}, f_k)|}^2 E_k^{(\rm d)}}{ \beta(0,0) \frac{\mathrm{N}_0}{D_2} E_k^{(\rm d)} + \beta(0,0) \frac{\mathrm{N}_0 T_{\rm cs}}{T_{\rm s}} E^{(\rm r)} e^{- \mathrm{Var}\{\theta_{\rm L}(t)\}} + \frac{{[\mathrm{N}_0]}^2 T_{\rm cs}}{D_2 T_{\rm s}} } \label{eqn_avgSNR_exprn} \\
{\rm iSE}^{\rm PACE} \big(\mathbf{H}(t) \big) & \gtrapprox & \sum_{k \in \mathcal{K}} \frac{1}{K} \log \left(1 \!+\! \frac{ M_{\rm rx} E^{(\rm r)} e^{-\mathrm{Var}\{\theta_{\rm L}(t)\}} {|\beta(f_{\rm c}, f_k)|}^2 E_k^{(\rm d)}}{ \beta(0,0) \frac{\mathrm{N}_0}{D_2} E_k^{(\rm d)} \!\!+\! \beta(0,0) \frac{\mathrm{N}_0 T_{\rm cs}}{T_{\rm s}} E^{(\rm r)} e^{- \mathrm{Var}\{\theta_{\rm L}(t)\}} + \frac{{[\mathrm{N}_0]}^2 T_{\rm cs}}{D_2 T_{\rm s}}  } \right), \!\!\!\!\!\!\!\! \label{eqn_avg_cap}
\end{eqnarray}
\end{subequations}
\end{figure*}
where $\beta(\dot{f},\ddot{f}) = \sum_{\ell=0}^{L-1} {|\alpha_{\ell}|}^2 {|{\mathbf{a}_{\rm tx}(\ell) }^{\dag} \mathbf{t}|}^2 e^{{\rm j} [2 \pi \dot{f} (\hat{\tau}_{\ell}-\tau_{\ell}) - 2 \pi \ddot{f} \tau_{\ell}]}$ and $\gtrapprox$ represents a $\geq$ inequality at a high enough SNR such that the approximations in Section \ref{sec_analog_phase_amp_est} are accurate.
\end{theorem}
\begin{proof}
Substituting \eqref{eqn_integral_phase_2} in \eqref{eqn_exprn_Y_k}, and by treating the received signal component corresponding to $\hat{\mathbf{W}}^{(\rm r)}$, i.e., ${[\hat{\mathbf{W}}^{(\rm r)}]}^{\dag} \boldsymbol{\mathcal{H}}(f_k) \mathbf{t}  x_k$, as noise, we can obtain a lower bound to the mean SNR as:
\begin{flalign}
& \gamma^{\rm PACE}_{k}(\mathbf{H}(t)) & \nonumber \\
& \qquad \gtrapprox \mathbb{E}_{\hat{\mathbf{W}}^{(\rm r)}} \Bigg\{ \frac{T_{\rm cs} {\big|\mathbf{t}^{\dag} \boldsymbol{\hat{\mathcal{H}}}(f_0)^{\dag} \boldsymbol{\mathcal{H}}(f_k) \mathbf{t} \big|}^2 E_k^{(\rm d)} E^{(\rm r)} e^{-\mathrm{Var}\{\theta_{\rm L}(t)\}}}{ \mathbb{E}_{x_k,\tilde{\mathbf{W}}^{\rm d}[k]} \Big\{ {\big|\mathbf{I}_{\rm PACE}^{\dag} \tilde{\mathbf{W}}^{\rm d}[k] + {[\hat{\mathbf{W}}^{(\rm r)}]}^{\dag} \boldsymbol{\mathcal{H}}(f_k) \mathbf{t}  x_k \big|}^2 \Big\} } \Bigg\} & \nonumber \\
& \qquad \stackrel{(1)}{\geq} \frac{T_{\rm cs} {\big|\mathbf{t}^{\dag} \boldsymbol{\hat{\mathcal{H}}}(f_0)^{\dag} \boldsymbol{\mathcal{H}}(f_k) \mathbf{t} \big|}^2 E_k^{(\rm d)} E^{(\rm r)} e^{- \mathrm{Var}\{\theta_{\rm L}(t)\}}}{ \mathbb{E}_{\hat{\mathbf{W}}^{(\rm r)}} \Big\{ {\|\mathbf{I}_{\rm PACE}\|}^2 \mathrm{N}_0 T_{\rm cs}/T_{\rm s} + {\big|{[\hat{\mathbf{W}}^{(\rm r)}]}^{\dag} \boldsymbol{\mathcal{H}}(f_k) \mathbf{t}\big|}^2 E_k^{(\rm d)}  \Big\} } & \nonumber \\
& \qquad = \frac{{\big|\mathbf{t}^{\dag} \boldsymbol{\hat{\mathcal{H}}}(f_0)^{\dag} \boldsymbol{\mathcal{H}}(f_k) \mathbf{t} \big|}^2 E_k^{(\rm d)} E^{(\rm r)} e^{- \mathrm{Var}\{\theta_{\rm L}(t)\}}}{\text{denom.}} & \nonumber \\
&\text{denom.} = \bigg[ \mathbf{t}^{\dag} \boldsymbol{\hat{\mathcal{H}}}(f_0)^{\dag} \boldsymbol{\hat{\mathcal{H}}}(f_0) \mathbf{t} E^{(\rm r)} e^{- \mathrm{Var}\{\theta_{\rm L}(t)\}} & \nonumber \\
& \qquad \qquad \qquad +\! \frac{M_{\rm rx} \mathrm{N}_0}{D_2} \bigg] \frac{\mathrm{N}_0 T_{\rm cs}}{T_{\rm s}} \!+\! \frac{\mathrm{N}_0}{D_2} \mathbf{t}^{\dag} {\boldsymbol{\mathcal{H}}(f_k)}^{\dag} \boldsymbol{\mathcal{H}}(f_k) \mathbf{t} E_k^{(\rm d)} & \nonumber \\
&\qquad \stackrel{(2)}{=} \frac{M_{\rm rx}{|\beta(f_{\rm c}, f_k)|}^2 E_k^{(\rm d)} E^{(\rm r)} e^{- \mathrm{Var}\{\theta_{\rm L}(t)\}}}{ \beta(0, 0) \frac{\mathrm{N}_0 T_{\rm cs}}{T_{\rm s}} E^{(\rm r)} e^{- \mathrm{Var}\{\theta_{\rm L}(t)\}} \!+\! \frac{\mathrm{N}^2_0 T_{\rm cs}}{D_2 T_{\rm s}} \!+\! \frac{\mathrm{N}_0}{D_2} \beta(0, 0) E_k^{(\rm d)} }, \!\!\!\!\!\!\!\!\!\!\!\! &  \label{eqn_theorem_proof1}
\end{flalign}
where $\stackrel{(1)}{\geq}$ follows from the Jensen's inequality and $\stackrel{(2)}{=}$ from the orthogonality of the array response vectors. Similarly, by treating ${[\hat{\mathbf{W}}^{(\rm r)}]}^{\dag} \boldsymbol{\mathcal{H}}(f_k) \mathbf{t} x_k$ as Gaussian noise independent of $x_k$, a lower bound on the mean iSE can be obtained as:
\begin{flalign}
&{\rm iSE}^{\rm PACE} \big(\mathbf{H}(t) \big) \gtrapprox \mathbb{E}_{\hat{\mathbf{W}}^{(\rm r)}} \sum_{k \in \mathcal{K}} \frac{1}{|\mathcal{K}|} \log\bigg[ 1 & \nonumber \\
& \qquad \qquad + \frac{T_{\rm cs} {\big|\mathbf{t}^{\dag} \boldsymbol{\hat{\mathcal{H}}}(f_0)^{\dag} \boldsymbol{\mathcal{H}}(f_k) \mathbf{t} \big|}^2 E_k^{(\rm d)} E^{(\rm r)} e^{- \mathrm{Var}\{\theta_{\rm L}(t)\}}}{ \mathbb{E}_{x_k,\tilde{\mathbf{W}}^{\rm d}[k]} \Big\{ {\big|\mathbf{I}_{\rm PACE}^{\dag} \tilde{\mathbf{W}}^{\rm d}[k] + {[\hat{\mathbf{W}}^{(\rm r)}]}^{\dag} \boldsymbol{\mathcal{H}}(f_k) \mathbf{t}  x_k \big|}^2 \Big\} } \bigg] & \nonumber \\
& \qquad \geq \sum_{k \in \mathcal{K}} \!\!\frac{1}{|\mathcal{K}|} \log\bigg[ 1 & \nonumber \\
& \qquad \qquad \!+\! \frac{M_{\rm rx}{|\beta(f_{\rm c}, f_k)|}^2 E_k^{(\rm d)} E^{(\rm r)} e^{- \mathrm{Var}\{\theta_{\rm L}(t)\}}}{ \beta(0, 0) \frac{\mathrm{N}_0 T_{\rm cs}}{T_{\rm s}} E^{(\rm r)} e^{- \mathrm{Var}\{\theta_{\rm L}(t)\}} \!+\! \frac{\mathrm{N}^2_0 T_{\rm cs}}{D_2 T_{\rm s}} \!+\! \frac{\mathrm{N}_0}{D_2} \beta(0, 0) E_k^{(\rm d)} } \bigg], & \nonumber
\end{flalign}
where we use similar steps to \eqref{eqn_theorem_proof1}.
\end{proof}
%
The array response orthogonality condition in Theorem \ref{Theorem1} is satisfied if the scatterers corresponding to different MPCs are well separated and $M_{\rm rx} \gg L$ \cite{Ayach2012}. Note that even though the RX does not explicitly estimate the array response vectors $\mathbf{a}_{\rm rx}(\ell)$ for the MPCs, we still observe an RX beamforming gain of $M_{\rm rx}$ in \eqref{eqn_avgSNR_exprn}. The impact of imperfect MRC combining and the resulting frequency-selective fading is quantified by $\beta(f_{\rm c}, f_k)$, where note that $|\beta(f_{\rm c}, f_k)| \leq |\beta(0,0)|$. Another drawback of the fading is that it may cause a drastic drop in performance of the one PLL architecture in Section \ref{subsec_pll_analysis} if $|A_1^{(\rm r)}|$ - the reference signal strength at the antenna $1$ - falls in a fading dip, as is evident from \eqref{eqn_PN_var} and \eqref{eqn_lemma_perf_anal}. Note however that the weighted arraying architecture in Section \ref{subsec_carrier_arraying} enjoys \emph{diversity} against such fading by recovering the reference tone from multiple antennas $\mathcal{M}$.

\section{Initial access and aCSI estimation at the BS} \label{sec_IA_aCSI_at_BS}
In this section we suggest how aCSI can be acquired at the BS during the TX beamformer design phase and also propose a sample IA protocol that can utilize PACE. Note that power allocation, user-scheduling and design of the TX beamformer $\mathbf{t}$ requires knowledge of the TX array response vectors and amplitudes $\{|\alpha_{\ell}|, \mathbf{a}_{\rm tx}(\ell)\}$ for the different UEs. 
Such aCSI can be acquired at the BS either via uplink CE, or by downlink CE with CSI feedback from the RX. Uplink CE can be performed by transmitting an orthogonal pilot from each UE omni-directionally, and using any of the digital CE algorithms from Section \ref{sec_intro} at the BS. Note that PACE cannot be used at the BS since the pilots from multiple UEs need to be separated via digital processing. For downlink CE with feedback, the BS transmits reference signals sequentially along different transmit precoder beams (beam sweeping), with $D$ reference symbols for each beam. The UEs perform PACE for each TX beam, and provide the BS with uplink feedback about the corresponding link strength for data transmission. 

The suggested IA protocol is somewhat similar to the downlink CE with feedback, where the BS performs beam sweeping along different angular directions, possibly with different beam widths. For each TX beam, the BS transmits $D$ reference symbols, followed by a sequence of primary (PSS) and secondary synchronization sequences (SSS). 
The RX performs PACE, and provides uplink feedback to the BS upon successfully detecting a PSS. However due to lack of prior timing synchronization during IA phase, the `filter, sample and hold' circuit in Section \ref{subsec_phase_amp_estimate} cannot be used directly for the PACE. One alternative is to allow continuous transmission of the reference tone even during the PSS and SSS with the following suggested symbol structure:
\begin{eqnarray}
\tilde{\mathbf{s}}^{(\rm ia)}_{\rm tx}(t) &=& \sqrt{\frac{2}{T_{\rm cs}}} \mathbf{t}\bigg[ \sqrt{E^{(\rm r)}} + \sum_{k \in \mathcal{K}\setminus \mathcal{G}} x^{(\rm p)}_k e^{{\rm j} 2 \pi f_k t} \bigg] e^{{\rm j} 2 \pi f_{\rm c} t } \label{eqn_tx_signal_pilot}
\end{eqnarray}
where $\mathcal{G}$ defines a guard band around the reference tone, to reduce the impact of the data sub-carriers on the PLL output. The amplitude and phase estimation can then be performed similar to Section \ref{subsec_phase_amp_estimate}, by multiplying the received signal at each antenna with the PLL output and then filtering with a low pass filter with cut-off frequency $1/(D_2 T_{\rm cs})$. Due to the continuous availability of the reference tone, the filter outputs can be directly used to control the phase shifter at each antenna without the `sample and hold' operation. 
Since $D = {\rm O}(1)$, the IA latency does not scale with $M_{\rm rx}$ and yet the PSS/SSS symbols can exploit the RX beamforming gain, thus improving cell discovery radius and/or reducing IA overhead. 

\section{Simulation Results} \label{sec_sim_results}
For the simulation results, we consider a single cell scenario with a $\lambda/2$-spaced $32 \times 8$ ($M_{\rm tx} = 256$) antenna BS and one representative UE with a $\lambda/2$-spaced $16 \times 4$ ($M_{\rm rx} = 64$) antenna array, having one down-conversion chain and using PACE aided beamforming. The BS has perfect aCSI and transmits one spatial OFDM data stream to this UE with $K=1024$ sub-carriers and the beamformer $\mathbf{t}$ aligned with the strongest channel MPC. 
The RX beamformer design phase is assumed to last $D=6$ symbols with $D_2 = 2$, where the BS transmits reference symbols with power $E^{(\rm r)} = 20 E_{\rm cs}/K$ (to satisfy spectral mask regulations). The system parameters for the one PLL and weighted arraying case, respectively, are as given in Table \ref{Table_sim_param} on page \pageref{Table_sim_param}. For comparison to existing schemes, we include the performance of RTAT - the continuous ACE based beamforming scheme in \cite{Ratnam_ICC2018}, and of  statistical RX analog beamforming \cite{Sudarshan}, where the beamformer is the largest eigen-vector of the RX spatial correlation matrix: $\mathbf{R}_{\rm rx}(\mathbf{t}) = \frac{1}{K} \sum_{k \in \mathcal{K}} \boldsymbol{\hat{\mathcal{H}}}(f_k) \mathbf{t} \mathbf{t}^{\dag} {\boldsymbol{\hat{\mathcal{H}}}(f_k)}^{\dag}$. For both these schemes we ignore impact of phase noise and additionally, for statistical beamforming we consider two cases: (a) perfect knowledge of $\mathbf{R}_{\rm rx}(\mathbf{t})$ at the RX and (b) estimate of $\mathbf{R}_{\rm rx}(\mathbf{t})$ obtained using sparse-ruler sampling \cite{Pal2010} - a reduced complexity digital CE technique. Note that PACE uses $6$ reference symbols per beamformer update phase, RTAT avoids reference symbols but requires continuous transmission of the reference and sparse-ruler sampling requires $21$ pilot symbols for $M_{\rm rx} = 64$. 

We first consider a sparse multi-path channel having $L=3$ MPCs with delays $\hat{\tau}_{\ell}= \{0, 20, 40\}$ns, angles of arrival $\psi^{\rm rx}_{\rm azi} = \{0, \pi/6, -\pi/6\}$, $\psi^{\rm rx}_{\rm ele} = \{0.45\pi, \pi/2, \pi/2\}$ and effective amplitudes $\frac{\alpha_{\ell} \mathbf{a}_{\rm tx}(\ell)^{\dag} \mathbf{t}}{\sqrt{\beta(0,0)}} = \{\sqrt{0.6}, -\sqrt{0.3}, \sqrt{0.1}\}$, respectively, during the RX beamformer design phase and $\tau_{\ell} = \hat{\tau}_{\ell} + \{30, 25, 25\}$ps for one snapshot of the data transmission phase. For this channel, the mean iSE of PACE aided beamforming, obtained using Monte-Carlo simulations with the non-linear PLL equations \eqref{eqn_PLL_diff_eqn_1}, \eqref{eqn_carr_array_sec}, \eqref{eqn_carr_array_prim}, is compared to the analytical approximation \eqref{eqn_avg_cap}, and the performance other schemes in Fig.~\ref{Fig_iSE_compare_sparse}. Since the RX beamformer $\mathbf{I}_{\rm PACE}$ in \eqref{eqn_integral_phase_1} is random, the one sigma interval of iSE is also depicted as a shaded region here. As is evident from the results, the beamforming gain with PACE aided beamforming is only $2$ dB lower than that of statistical beamforming, above a certain SNR threshold. Below this threshold, however, PACE experiences an exponential decay in peformance due to the oscillator phase noise, as also predicted by Theorem \ref{Theorem1}. As is expected, this SNR threshold is lower for weighted carrier arraying than for one PLL. Furthermore, the derived analytical approximations are also accurate above this SNR threshold. PACE also outperforms RTAT at high SNR due to the judicious transmission of the reference, while the deceptively better performance of RTAT at low SNR is due to neglect of phase-noise. Note that these PACE results are obtained for an oscillator offset of $5$ MHz (see Table \ref{Table_sim_param} on page \pageref{Table_sim_param}). Better performance can be achieved if the PLL is optimized for more accurate local oscillators. 
\begin{figure}[h] 
\centering
\subfloat[Influence of SNR]{\includegraphics[width= 0.45\textwidth]{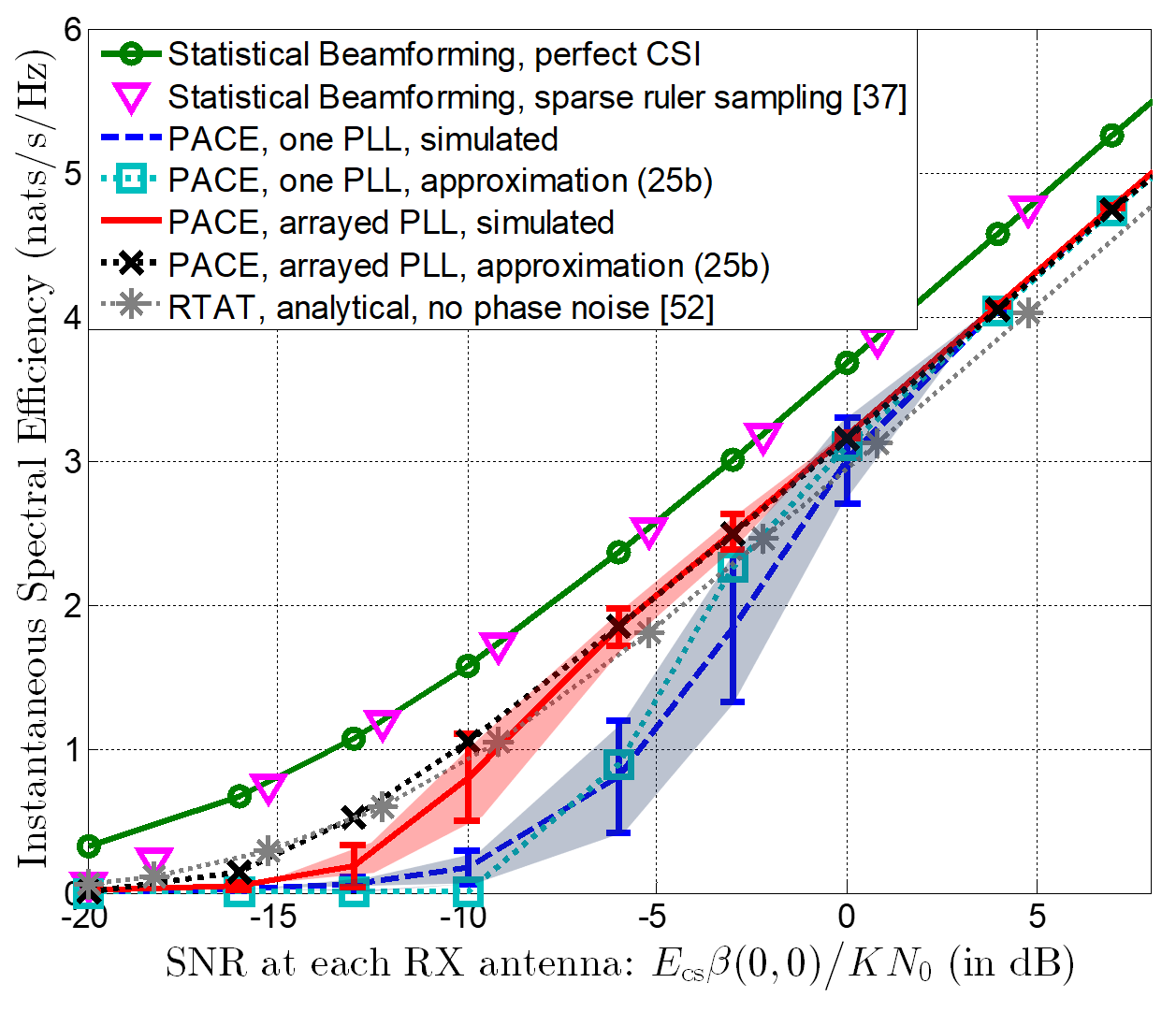} \label{Fig_iSE_compare_sparse}} \hspace{5 mm}
\subfloat[Influence of number of MPCs]{\includegraphics[width= 0.45\textwidth]{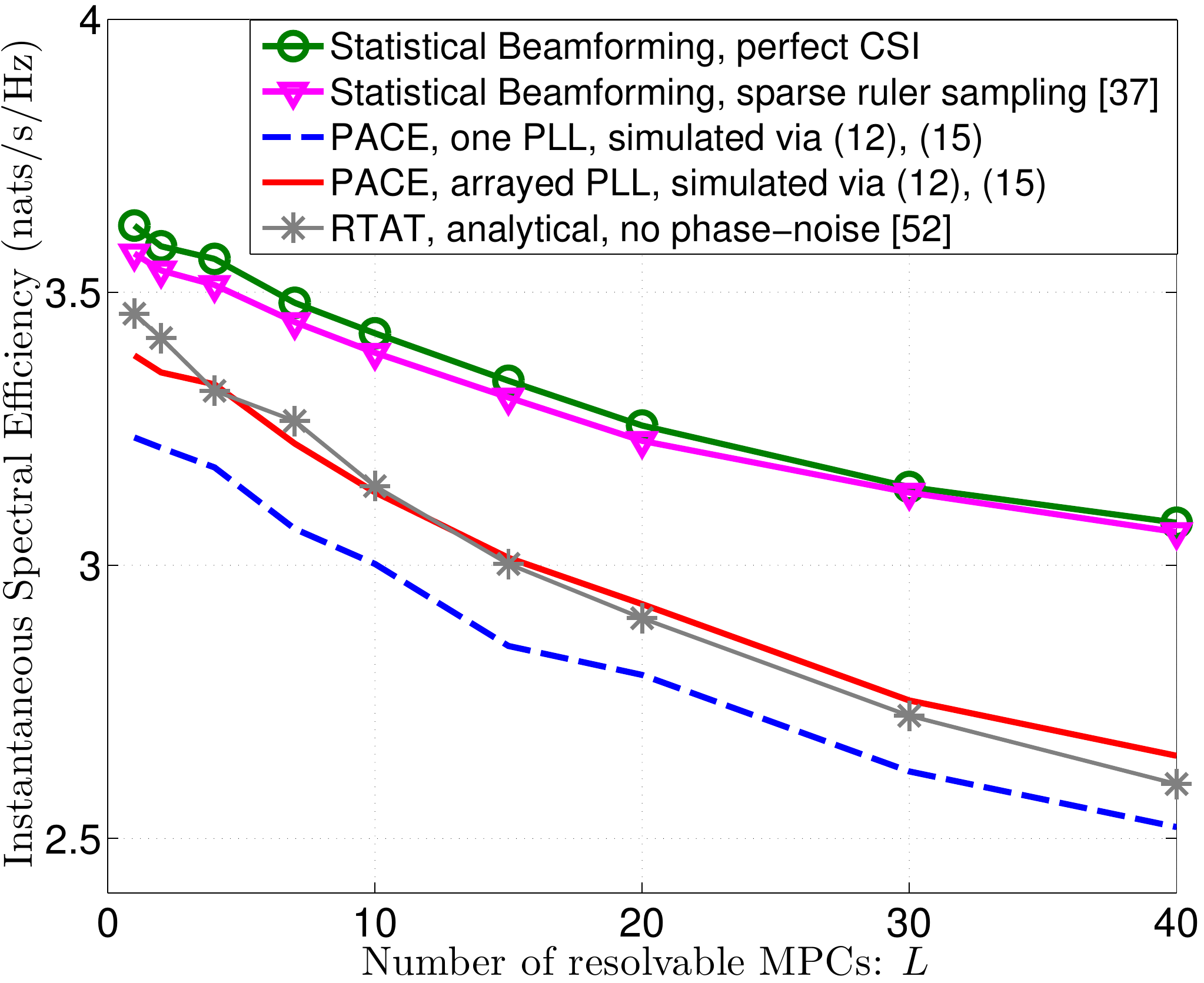} \label{Fig_iSE_compare_dense}}
\caption{Comparison of iSE for PACE based beamforming and other schemes versus SNR and number of MPCs. Here $E^{(\rm r)} = 20 E_{\rm cs}/K$, $E^{(\rm d)}_k = E_{\rm cs}/K$ $\forall k \in \mathcal{K}$ and the PLL parameters are from Table \ref{Table_sim_param} on page \pageref{Table_sim_param}. For Fig.~\ref{Fig_iSE_compare_dense} we use $\frac{ \beta(0,0) E_{\rm cs} }{K \mathrm{N}_0} = 1$} 
\label{Fig_iSE_compare} 
\end{figure} 

To study the impact of more realistic channels and number of MPCs, we next model the channel as a rich scattering stochastic channel with $L$ resolvable MPCs, each with $10$ unresolved sub-paths. Here the MPCs and sub-paths are generated identically to the clusters and rays, respectively, in the 3GPP TR38.900 Rel 14 channel model (UMi NLoS scenario) \cite{TR38900_chanmodel}. The only difference from \cite{TR38900_chanmodel} is that we use an intra-cluster delay spread of $1 ns$ and an intra-cluster angle spread of $\pi/50$ (for all elevation, azimuth, arrival and departure), to ensure that the sub-paths of each MPC are unresolvable. The channel SNR at each RX antenna (including the TX beamforming gain) is fixed at $0$ dB, and the channel variation between beamformer design phase and one snapshot of the data transmission phase is modeled by assuming that the RX moves a distance of $d = 2$cm in a random azimuth direction without changing its orientation. Note that this channel can also be represented by our system model by replacing $L$ in \eqref{eqn_channel_impulse_resp} with $10L$.
For this stochastic channel model, the mean iSE for PACE aided beamforming, averaged over channel realizations, is compared to RTAT and statistical beamforming in Fig.~\ref{Fig_iSE_compare_dense}. For computational tractability, we skip the non-linear PLL simulation and use the analytical expressions \eqref{eqn_integral_phase_2} and \eqref{eqn_cap} to quantify performance of PACE.\footnote{Note that \eqref{eqn_avg_cap} is not applicable due to non-orthogonality of array response vectors.} These expressions are accurate at $0$dB SNR as observed from Fig.~\ref{Fig_iSE_compare_sparse}. As observed from the results, the loss in beamforming gain for PACE aided beamforming increases with $L$, and therefore PACE is mainly suitable for channels with $L \leq 10$ resolvable MPCs. It must be emphasized that such cases may frequently occur at mm-wave frequencies, where the number of resolvable MPCs/clusters with significant energy (within $20$dB of the strongest) is on the order $3-10$ \cite{Akdeniz2014, TR38900_chanmodel}. 

Note that for the iSE results in this section, we did not include the CE overhead. While digital appoaches like sparse ruler sampling \cite{Pal2010, Caire2017} require $21$ pilots (for $M_{\rm rx}=64$), PACE uses only $D=6$ pilots. The corresponding overhead reduction is significant when downlink CE with feedback is used for aCSI acquisition at the BS, such as in frequency division duplexing systems.\footnote{Even in time division duplexing systems, dowlink CE with feedback may be used during the IA phase, causing a large IA latency. PACE can help reduce IA latency in such situations.} For example with exhaustive beamscanning \cite{Jeong2015} at the TX and an aCSI coherence time of $10$ms, the BS aCSI acquisition overhead reduces from $40\%$ for sparse ruler techniques to $11\%$ for PACE (see Section~\ref{sec_IA_aCSI_at_BS} for protocol). The overhead reduction is expected to be higher if the additional time required for beam switching and settling \cite{Sands2002, Venugopal2017} are also taken into account. Thus, PACE aided beamforming shows potential in solving the CE overhead issue of hybrid massive MIMO systems, with minimal degradation in performance. 

\section{Conclusions} \label{sec_conclusions}
This paper proposes the use of PACE for designing the RX beamformer in massive MIMO systems. This process involves transmission of a reference sinusoidal tone during each beamformer design phase, and estimation of its received amplitude and phase at each RX antenna using analog hardware. A one PLL based carrier recovery circuit is proposed to enable the PACE receiver, and its analysis suggests that the quality of obtained channel estimates decay exponentially with inverse of the SNR at the PLL input. To remedy this and also to obtain diversity against fading, a multiple PLL based weighted carrier arraying architecture is also proposed. The performance analysis suggests that PACE aided beamforming can be interpreted as using the channel estimates on one sub-carrier to perform beamforming on other sub-carriers, with an additional loss factor corresponding to the circuit phase-noise. Simulation results suggest that PACE aided beamforming suffers only a small beamforming loss in comparison to conventional analog beamforming in sparse channels, at sufficiently high SNR. This loss however increases with the number of channel MPCs $L$, and hence PACE is mostly suitable for sparse channels with few MPCs. The CE overhead reduction with PACE is significant when downlink CE with feedback is required. Benefits of PACE aided beamforming during IA phase are also discussed, although a more detailed analysis will be a subject for future work. Similarly the performance of PACE at very low SNR and with system mismatches/imperfections also requires more attention. 

\section*{Acknowledgement}
The authors would like to thank Dr. W. C. Lindsey and Dr. H. Hashemi of University of Southern California for their helpful comments regarding phase-locked loops.

\bibliographystyle{ieeetr}
\bibliography{references}

\begin{IEEEbiography}[{\includegraphics[width=1in,height=1.25in,clip,keepaspectratio]{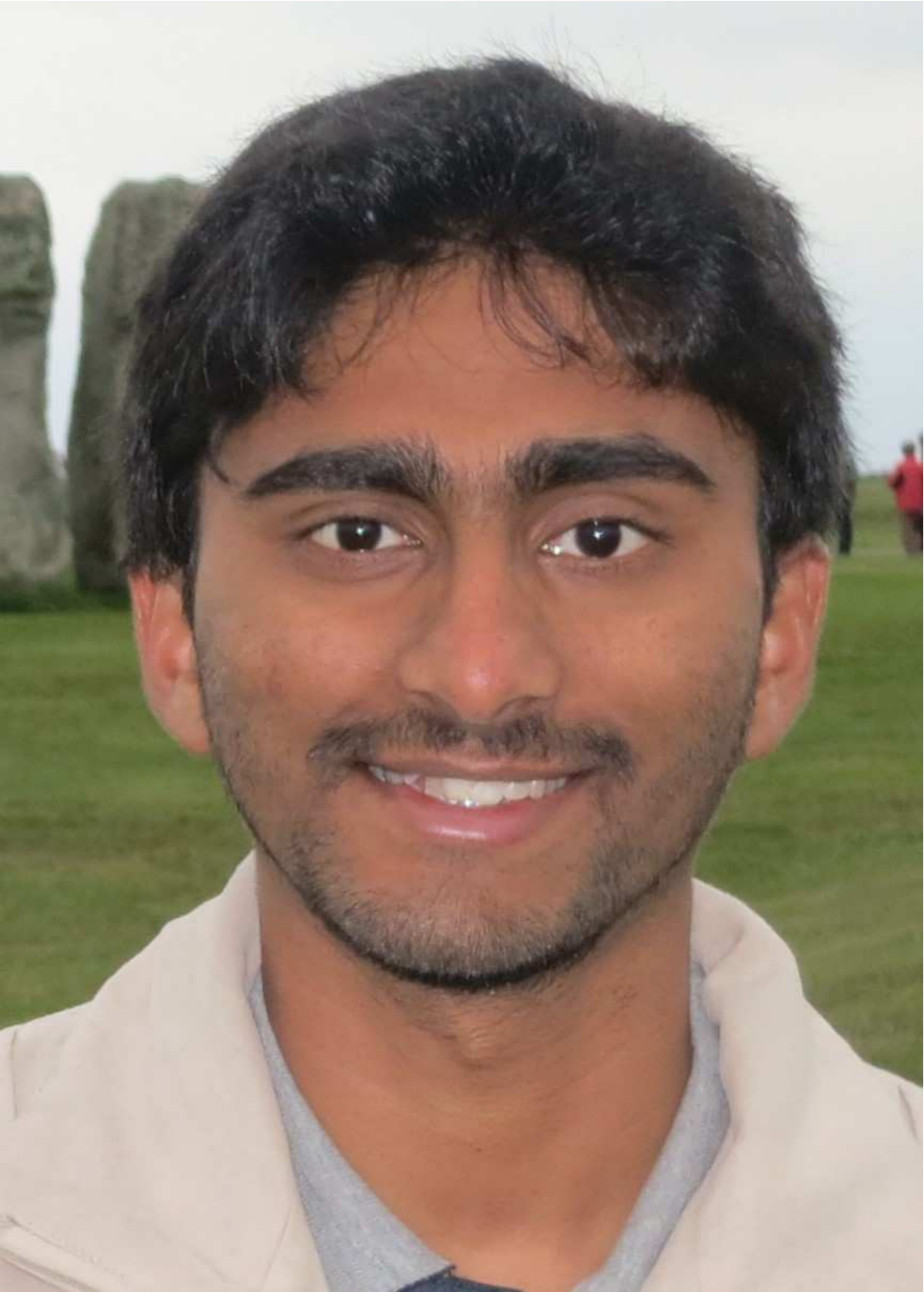}}]{Vishnu V. Ratnam}
(S'10) received the B.Tech. degree (Hons.) in electronics and electrical communication engineering from IIT Kharagpur, Kharagpur, India
in 2012, where he graduated as the Salutatorian for the class of 2012. He received the Ph.D. degree in electrical engineering from University of Southern California, Los Angeles, CA, USA in 2018. He is currently a senior research engineer in the Standards and Mobility Innovation Lab at Samsung Research America. His research interests are in reduced complexity transceivers for large antenna systems (massive MIMO/mm-wave) and ultrawideband systems, channel estimation techniques, manifold signal processing and in resource allocation problems in multi-antenna networks. 

Mr. Ratnam is a recipient of the Best Student Paper Award at the IEEE International Conference on Ubiquitous Wireless Broadband (ICUWB) in 2016, Bigyan Sinha memorial award in 2012 and is a member of the Phi-Kappa-Phi honor society.
\end{IEEEbiography}

\begin{IEEEbiography}[{\includegraphics[width=1in,height=1.25in,clip,keepaspectratio]{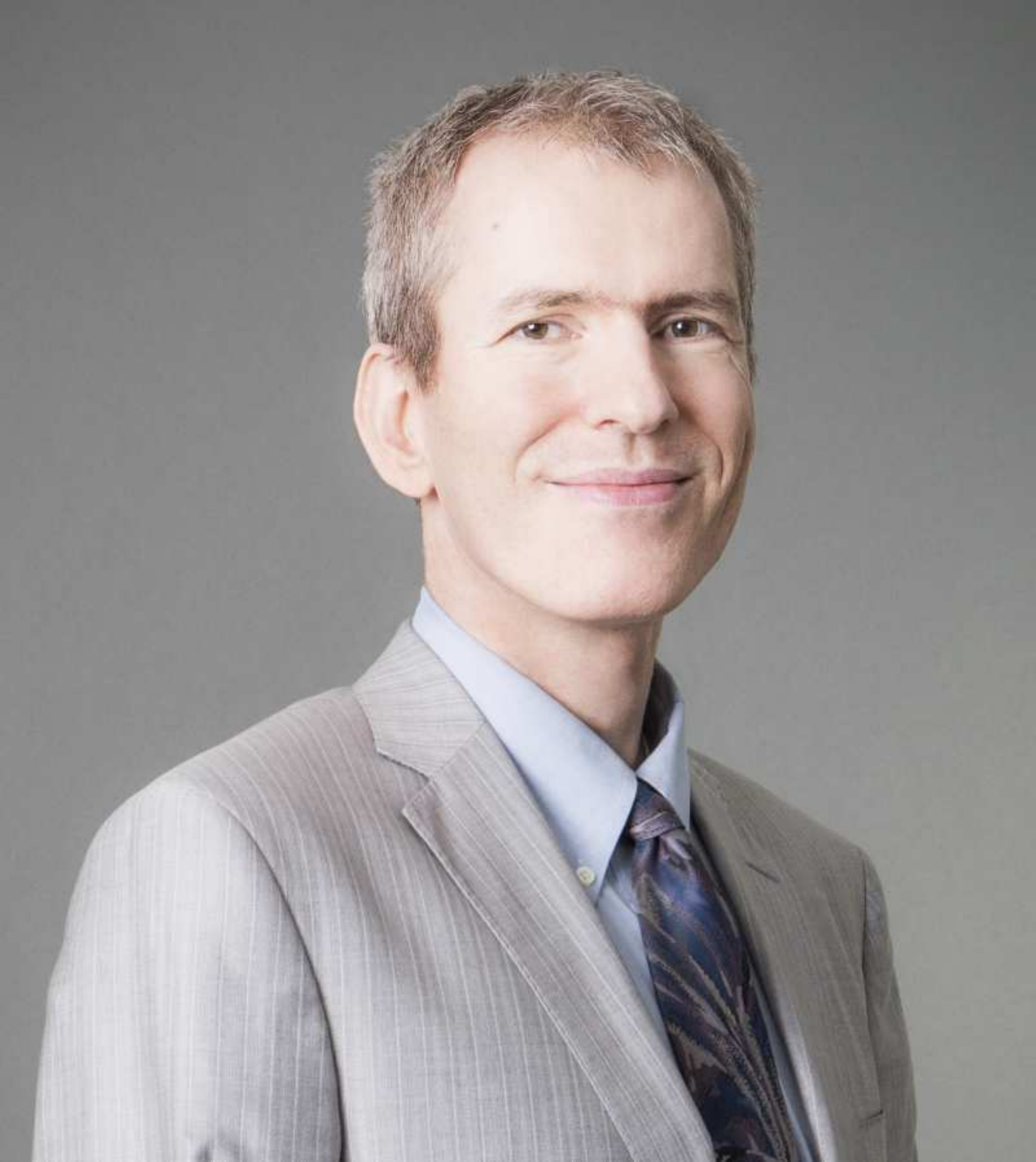}}]{Andreas F. Molisch}
(S'89--M'95--SM'00--F'05) received the Dipl. Ing., Ph.D., and habilitation degrees from the Technical University of Vienna, Vienna, Austria, in 1990, 1994, and 1999, respectively. He subsequently was with AT\&T (Bell) Laboratories Research (USA); Lund University, Lund, Sweden, and Mitsubishi Electric Research Labs (USA). He is now a Professor and the Solomon Golomb -- Andrew and Erna Viterbi Chair at the University of Southern California, Los Angeles, CA, USA. 

His current research interests are the measurement and modeling of mobile radio channels, multi-antenna systems, wireless video distribution, ultra-wideband communications and localization, and novel modulation formats. He has authored, coauthored, or edited four books (among them the textbook Wireless Communications, Wiley-IEEE Press), 19 book chapters, more than 230  journal papers, more than 320 conference papers, as well as more than 80 patents and 70 standards contributions.

Dr. Molisch has been an Editor of a number of journals and special issues, General Chair, Technical Program Committee Chair, or Symposium Chair of multiple international conferences, as well as Chairman of various international standardization groups. He is a Fellow of the National Academy of Inventors, Fellow of the AAAS, Fellow of the IET, an IEEE Distinguished Lecturer, and a Member of the Austrian Academy of Sciences. He has received numerous awards, among them the Donald Fink Prize of the IEEE, the IET Achievement Medal, the Armstrong Achievement Award of the IEEE Communications Society, and the Eric Sumner Award of the IEEE.
\end{IEEEbiography}

\end{document}